\documentclass[11pt]{article}

\usepackage[margin=1in]{geometry}
\usepackage{amsmath,amssymb,amsfonts,amsthm,mathtools,bm}
\usepackage{booktabs}
\usepackage{physics}
\usepackage{enumitem}
\usepackage{xcolor}
\usepackage[colorlinks=true,linkcolor=blue,citecolor=blue,urlcolor=blue]{hyperref}
\usepackage[nameinlink,capitalise]{cleveref}

\usepackage{algorithm}
\usepackage{algpseudocode}
\usepackage{quantikz}

\newtheorem{theorem}{Theorem}[section]
\newtheorem{lemma}[theorem]{Lemma}

\newtheorem{remark}[theorem]{Remark}
\newtheorem{corollary}[theorem]{Corollary}

\crefname{assumption}{Assumption}{Assumptions}
\Crefname{assumption}{Assumption}{Assumptions}
\crefname{corollary}{Corollary}{Corollaries}
\Crefname{corollary}{Corollary}{Corollaries}
\crefname{lemma}{Lemma}{Lemmas}
\Crefname{lemma}{Lemma}{Lemmas}
\crefname{remark}{Remark}{Remarks}
\Crefname{remark}{Remark}{Remarks}

\newcommand{\R}{\mathbb R}
\newcommand{\C}{\mathbb C}

\newcommand{\eps}{\varepsilon}

\newcommand{\polylog}{\operatorname{polylog}}
\newcommand{\bigO}{\mathcal{O}}
\newcommand{\widetildeO}{\widetilde{\mathcal{O}}}
\newcommand{\res}{\mathrm{res}}

\title{A Residual-Based Quantum Linear System Algorithm with Dynamic Stopping
and Applications to Elliptic PDEs}
\author{Xiantao Li \\ Xiantao.Li@psu.edu \\ Department of Mathematics, The Pennsylvania State University }

\begin{document}
\maketitle
\begin{abstract}
Quantum linear-system algorithms (QLSAs) have rigorous worst-case complexity guarantees, but their
runtimes are often chosen from spectral information assumed in advance. What is largely lacking is
an a posteriori progress flag: most QLSA workflows, unlike the classical counterparts,  do not provide a built-in mechanism to signal whether a particular instance has already converged.

For discretizations of elliptic PDEs $-\nabla\cdot(a(x)\nabla u(x))=f(x),$ with divergence--gradient structure
\[
   -\nabla\cdot \big(a(x)\nabla) \approx  A_h=G_h^\dagger G_h,
\]
we formulate a stable first-order ODE whose limiting solution block is the desired Galerkin solution. The
PDE-dependent scale is then \(\norm{G_h}=\bigO(h^{-1})\), comparable to factorized QLSA
constructions with square-root condition-number scaling. We design an augmented dynamics with
residual variables, in which measuring a residual register gives an on-the-fly convergence indicator without
reconstructing the solution vector. For smooth right-hand sides, dynamic stopping can reduce the
evolution time and gate count relative to a fixed worst-case schedule, and may also reduce exposure
to accumulated hardware errors. Numerical experiments for a two-dimensional finite element Poisson problem show that the
residual-register probability follows the actual error decay and, for some right-hand sides,
can stop the quantum circuit well before a conservative worst-case runtime estimate is reached.
\end{abstract}

\section{Introduction}

Quantum linear-system algorithms (QLSAs) have become one of the central primitives in
quantum algorithms. Since the HHL algorithm~\cite{HarrowHassidimLloyd2009}, substantial progress
has been made in improving the dependence on precision, condition number, and access models,
including variable-time amplitude amplification~\cite{Ambainis2012}, exponentially improved
precision dependence~\cite{ChildsKothariSomma2017}, block encodings and quantum singular value
transformation~\cite{GilyenSuLowWiebe2019}, and adiabatic or adiabatic-inspired
linear-system solvers~\cite{SubasiSommaOrsucci2019,AnLin2022,CostaAnSandersSuBabbushBerry2022}.
At the same time, QLSAs are not merely standalone algorithms. They are used as subroutines in
quantum algorithms for differential equations, least-squares problems, and other scientific-computing tasks.
Therefore, the practical behavior of a QLSA can directly affect the performance of a larger
quantum algorithm.

Most QLSA guarantees fix the algorithmic effort in advance, for example, a polynomial
degree, an adiabatic schedule length, or a simulation time, from worst-case spectral information
such as a condition-number bound or a promised singular-value interval.   While this is the right framework for proving black-box complexity theorems, it is
different from the way linear systems are solved in practice:
Classical iterative solvers
typically do not run only for a fixed worst-case iteration count. Instead, they monitor residuals
and stop when the current instance has converged. In this paper, we ask  whether a comparable
residual-based stopping mechanism can be incorporated into a QLSA.

This question is particularly natural for elliptic partial differential equations (PDEs), which are
a standard source of large sparse linear systems in scientific computing. After
discretization, a second-order elliptic operator usually gives a matrix \(A_h\) with condition
number
\begin{equation}\label{A-scale}
    \kappa(A_h)=\Theta(h^{-2}),
\end{equation}
where \(h\) represents the scale of the mesh size. This worst-case scale is determined by the highest-frequency modes
supported by the mesh. A particular right-hand side \(f\), however, may be smooth and thus have
little spectral weight in those modes. Consequently, the practical cost of solving the instance may
depend not only on \(\kappa(A_h)\), but also on the spectral content of the discrete vector
\(\bm b_h\). This type of right-hand-side dependence has begun to appear explicitly in the QLSA
literature:
for example, Li's instance-aware QLSA uses parameters of the augmented system
\(H=[A,-\bm b]\), together with a structure-aware right-rescaling, to obtain bounds that depend
on the alignment of \(\bm b\) with the linear system rather than only on \(\kappa(A)\)
\cite{Li2025BeyondCondition}. Our approach is complementary. We do not replace the worst-case
complexity theory. Rather, we formulate a QLSA as a time evolution and introduce a residual variable into the dynamics so that convergence
can be detected during the computation.

The elliptic problems considered here have a discrete divergence--gradient factorization
\[
    A_h=G_h^\dagger G_h,
\]
that comes from Galerkin projections in general. 
For finite differences and fixed-order finite elements, the first-order factor \(G_h\) is sparse and
has operator norm
\begin{equation}\label{G-scale}
    \norm{G_h}=\Theta(h^{-1})=
    \Theta\!\left(\sqrt{\kappa(A_h)}\right),
\end{equation}
whereas the stiffness matrix itself has scale \(\norm{A_h}=\Theta(h^{-2})\). This factorization can be leveraged for QLSAs. In particular, Orsucci and Dunjko studied
factorized positive-definite systems \(A=LL^\dagger\) and showed how the factor \(L\) can be used
to obtain square-root condition-number scaling~\cite{OrsucciDunjko2021}. In the elliptic setting,
\(L\) naturally emerges from the discrete divergence--gradient factor \(G_h\). 

Our first observation is that the same first-order scaling can be obtained from a stable ordinary
differential equation (ODE) formulation: the solution of the linear system appears as the limiting
solution block of the dynamics and the simulation complexity scales as $\bigO(h^{-1})$. It is worthwhile to mention that this dynamical viewpoint is familiar in classical numerical
analysis. Richardson iteration and related indirect methods for elliptic PDEs can be interpreted
as time discretizations of parabolic PDE dynamics whose steady state is the desired elliptic
solution~\cite[Chapter 8]{Varga2000MatrixIterative}. The distinct novelty here is that our ODE
formulation also carries residual variables. These variables are not part of the solution vector, but
they are essential for deciding whether the computation has converged.
In classical numerical linear algebra \cite{GolubVanLoan2013Matrix,Saad2003Iterative}, the residual
\[
    \bm r_h=\bm b_h-A_h\bm x_h
\]
is a basic diagnostic. It is tied to the algebraic error
\[
    \bm e_h=\bm x_\ast-\bm x_h
\]
through
\[
    A_h\bm e_h=\bm r_h .
\]

 Our
goal is to build an analogue of this mechanism into a quantum-compatible dynamical system. 
We formulate two related results. The first result shows that 
the elliptic problem  can be written as a stable first-order ODE whose quantum simulation scale is
\(\bigO(\sqrt{\kappa})\). The second result shows that the added residual register can be measured to enable
dynamic stopping.

\begin{theorem}[ODE-based QLSA]
\label{thm:intro-first-order-ode}
Let \(A_h=G_h^\dagger G_h\) be the normalized stiffness matrix for the elliptic discretization.
Introduce dynamic variables, 
\begin{equation}\label{var-aux}
    \bm w=
    \begin{bmatrix}
    \bm r\\
    \bm s
    \end{bmatrix},
    \qquad
    M_h=
    \begin{bmatrix}
    0&-G_h^\dagger\\
    G_h&-I
    \end{bmatrix},
    \qquad
    P_r=\begin{bmatrix}I&0\end{bmatrix},
\end{equation}
and consider the accumulator relaxation
\begin{equation}
\label{ode-relax}
    \dot{\bm x}=P_r\bm w,
    \qquad
    \dot{\bm w}=M_h\bm w,
    \qquad
    \bm x(0)=0,
    \qquad
    \bm w(0)=
    \begin{bmatrix}
    \bm b_h\\0
    \end{bmatrix}.
\end{equation}
Then \(\bm w(t)\) decays exponentially and \(\bm x(t)\) converges to the Galerkin solution
\(\bm x_\ast\) satisfying
\[
    G_h^\dagger G_h\bm x_\ast=\bm b_h .
\]
Suppose that \(G_h/\alpha_G\) admits sparse or block-encoding access with
\[
    \alpha_G=\Theta\bigl(\norm{G_h}\bigr),
\]
and that the right-hand-side state \(\ket{\bm b_h}\) can be prepared with cost \(C_b\). Then, for
any target accuracy \(\eps\), a stable quantum linear-ODE primitive with linear dependence on
\(\alpha_G T\) prepares an \(\eps\)-accurate normalized solution state with leading query and gate
complexity
\[
    \widetildeO\!\left(
        C_b+
        \alpha_G
        \polylog\frac{1}{\eps}
    \right),
\]
up to the standard output-normalization factor. Moreover,
after the relaxation time \(T=\widetildeO(\log(1/\eps))\), the \(\bm x\)-block has order-one
weight in the joint state, $(\bm w, \bm x)$, so extracting the solution block has constant overhead.

For elliptic discretizations with
\(
    \lambda_{\min}(A_h)=\Theta(1),
    \,
    \lambda_{\max}(A_h)=\Theta(h^{-2}),
    \,
    \norm{G_h}=\Theta(h^{-1}),
\)
this gives
\[
    \alpha_G=\Theta(h^{-1})
    =
    \Theta\!\left(\sqrt{\kappa(A_h)}\right),
\]
and hence the mesh-dependent complexity scaling is
\[
    \widetildeO\!\left(
        \sqrt{\kappa(A_h)}\,\polylog\frac{1}{\eps}
    \right).
\]
\end{theorem}

To run the dynamics in \cref{ode-relax}, suitable implementation primitives include QLSA-based ODE solvers, LCHS,
Schr\"odingerisation, Laplace-transform eigenvalue-transformation methods, and recent algorithms
for stable linear nonunitary dynamics
\cite{BerryChildsOstranderWang2017,Krovi2023,AnLiuLin2023,AnChildsLin2023,
JinLiuYu2024,AnChildsLinYing2024Laplace,LowSomma2025Nonunitary,li2025linear,JenningsLostaglioLowriePallisterSornborger2024}.

\medskip 

The estimate in
\cref{thm:intro-first-order-ode} is comparable to factorized QLSA constructions based on
square-root condition-number scaling ~\cite{OrsucciDunjko2021}. It is still, however, a worst-case statement. It does not
tell us whether a particular right-hand side requires the full worst-case polynomial degree, ODE
time, or simulation effort.
To illustrate this point, we consider the standard Poisson problem
\[
    -\Delta u=f
    \quad\text{in } \Omega=(0,1)^2,
    \qquad
    u=0
    \quad\text{on } \partial\Omega,
\]
discretized by continuous piecewise-linear finite elements on a uniform triangular mesh obtained by
splitting each square cell into two triangles. We test QLSAs for four choices of the right-hand side functions (detailed setup in \cref{sec:numerics}):
\[
\begin{alignedat}{2}
\mathrm{I:}\quad
& f_{\rm I}(x,y)
&&=2\pi^2\sin(\pi x)\sin(\pi y),\\
\mathrm{II:}\quad
& f_{\rm II}(x,y)
&&=2\pi^2\sin(\pi x)\sin(\pi y)
+\frac{5\pi^2}{2}\sin(2\pi x)\sin(\pi y),\\
\mathrm{III:}\quad
& f_{\rm III}(x,y)
&&=2\pi^2\sin(\pi x)\sin(\pi y)
+\frac{13\pi^2}{2}\sin(3\pi x)\sin(2\pi y)
+\frac{41\pi^2}{4}\sin(5\pi x)\sin(4\pi y),\\
\mathrm{IV:}\quad
& f_{\rm IV}(x,y)
&&=\exp\!\left(
-\frac{(x-1/2)^2+(y-1/2)^2}{0.1^2}
\right).
\end{alignedat}
\]
Cases I--III are increasingly oscillatory sine combinations, while Case IV is a localized smooth
bump. The matrix \(A_h\) and the first-order factor \(G_h\) are fixed by the mesh; only the 
right-hand side changes. These are standard manufactured-solution tests for elliptic discretizations
\cite{BrennerScott2008,Ciarlet1978}.

\begin{figure}[t]
\centering
\begin{minipage}{0.48\linewidth}
\centering
\includegraphics[width=\linewidth]{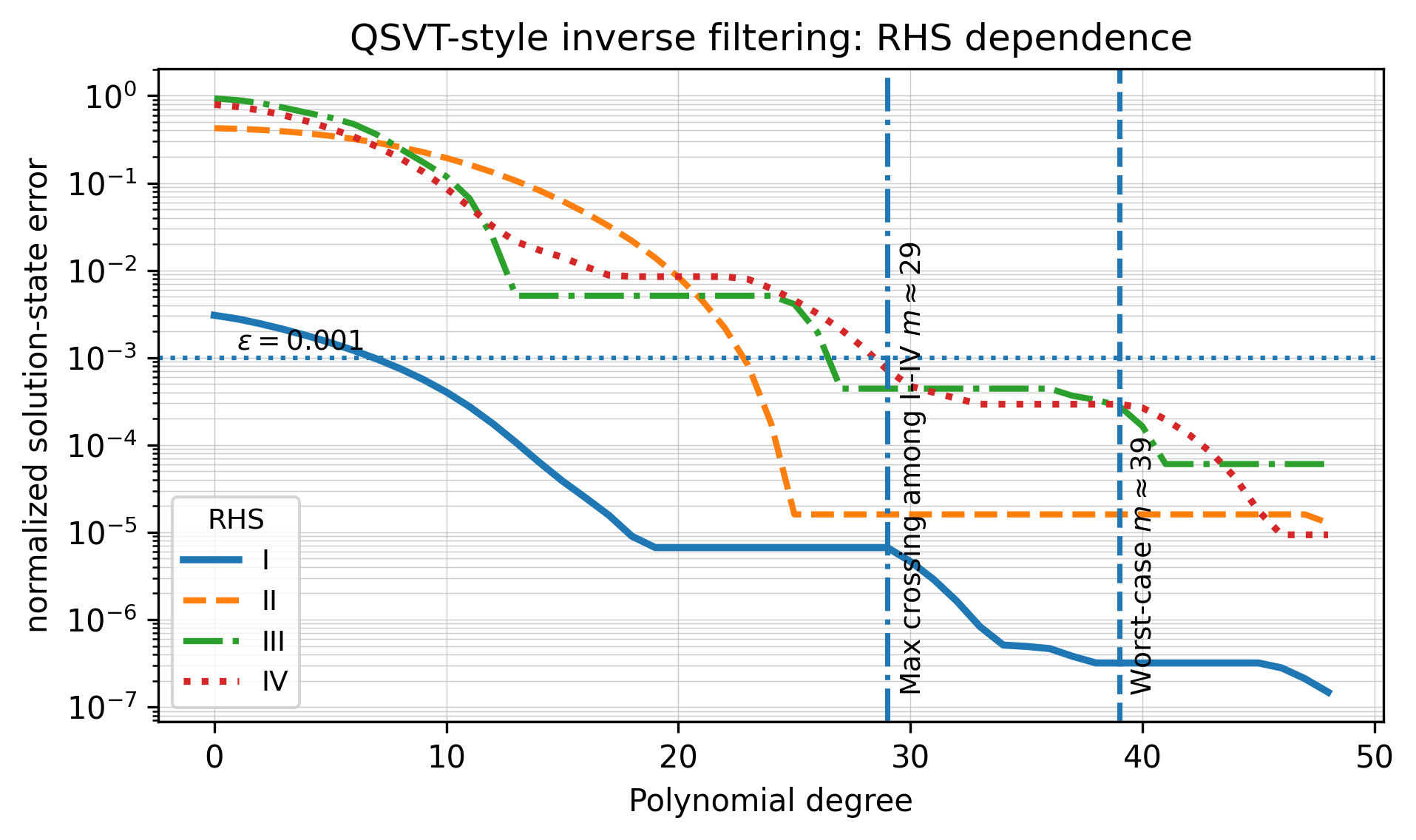}
\end{minipage}
\hfill
\begin{minipage}{0.48\linewidth}
\centering
\includegraphics[width=\linewidth]{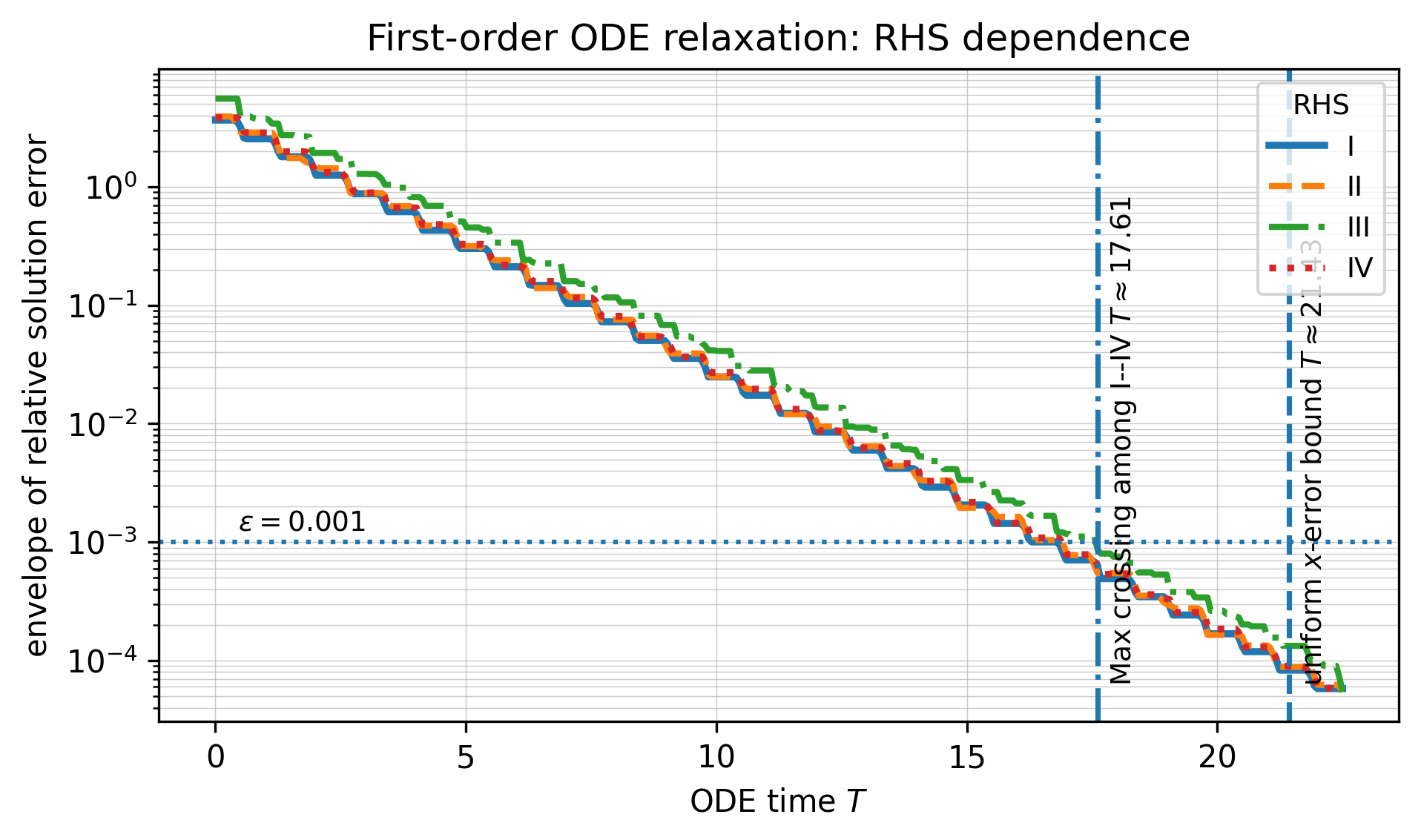}
\end{minipage}
\caption{Right-hand-side dependence for the finite element Poisson problem.
Left: a QSVT-style spectral inverse-polynomial experiment, where the horizontal axis is the
polynomial degree. This is a spectral-polynomial emulator for the inverse filter rather than a
hardware-level QSVT circuit. Right: the first-order ODE relaxation \cref{ode-relax},
where the horizontal axis is the ODE evolution time. The ODE curves use the 
relative solution error. In both panels, the dashed vertical line is the conservative worst-case
prediction, while the dash-dotted line marks the latest empirical crossing among the benchmark
cases I--IV.}
\label{fig:intro-rhs-dependence}
\end{figure}

In \cref{fig:intro-rhs-dependence}, the left panel implements a
QSVT-style inverse-polynomial experiment. Here QSVT refers to quantum singular
value transformation~\cite{GilyenSuLowWiebe2019}: given a block encoding of a
matrix, QSVT implements a polynomial transformation with a spectral inverse-polynomial
filter. We plot the normalized
solution-state error as a function of the polynomial degree. The right panel
shows the first-order ODE relaxation \cref{ode-relax} starting with zero initial condition. One observes that the practical degree or evolution time can depend
strongly on the right-hand side. Some input functions lead to solutions that cross the target tolerance  $\varepsilon=10^{-3}$ well before the conservative
worst-case line. Therefore, a fixed runtime can be overly conservative: the algorithm is guaranteed to succeed, but the circuit might be deep and the particular instance may have converged much earlier.

The ODE formulation in \cref{ode-relax} also has another
advantage that is not native to the standard one-shot QSVT inverse-filter workflow. It can naturally
incorporate an initial guess. A warm start can also be incorporated by choosing an initial \(\bm x(0)\) and an auxiliary
\(\bm q(0)\), then setting
\[
    \bm r(0)=\bm b_h-G_h^\dagger\bm q(0),
    \qquad
    \bm s(0)=G_h\bm x(0)-\bm q(0).
\]
If these residuals are small or concentrated in rapidly damped modes, the observed relaxation time
can be reduced.
The ODE formulation also makes it natural to organize residual readout through checkpoint circuits. Rather than estimating the
solution error directly, we monitor residual information carried by the dynamics. The next result
states the corresponding residual-augmented version of the first-order solver.

\begin{theorem}[ODE-based QLSA with dynamic stopping (Informal)]
\label{thm:intro-residual-based}
Under the hypotheses of \cref{thm:intro-first-order-ode},  and assume that the joint dynamics carrying the solution block
\(\bm x(t)\) and a mixed residual block  is implemented by a
quantum ODE solver.

Let \(t_{\rm ent}(\bm b_h)\) be the first checkpoint time at which a constant-threshold residual
test accepts, meaning that the measured residual probability indicates entry into the exponentially
decaying residual regime. Let
\[
    T_\star(\bm b_h,\eps)=t_{\rm ent}(\bm b_h)+\Delta t(\eps)
\]
be the final production time, where \(\Delta t(\eps)=\bigO(\log(1/\eps))\) is the deterministic
coasting time chosen from the stability estimate. Then a production run at time \(T_\star\)
prepares an \(\eps\)-accurate normalized solution state with leading cost
\[
    \widetildeO\!\left(
        C_b+
        \alpha_G T_\star(\bm b_h,\eps)\,
        \polylog\frac{1}{\eps}
    \right),
\]
up to the standard output-normalization factors of the chosen nonunitary ODE primitive. 
\end{theorem}

We will show that the checkpoint test uses a constant residual-probability threshold, independent of \(\eps\). Thus
the stopping test does not require resolving an \(\bigO(\eps^2)\)-probability event by direct
sampling.  The corresponding threshold and sampling model are elaborated in
\cref{subsec:constant-threshold-coasting,subsec:checkpoint-sampling,subsec:algorithmic-summary}.

\smallskip 
Overall, the residual-augmented accumulator ODE keeps the solution of the linear system as its limiting
\(\bm x\)-block, while adding a measurable indicator of progress. Thus the purpose of the residual
dynamics is not to improve the optimal worst-case black-box QLSA complexity. Rather, it is to make
the solver more adaptive: run the dynamics, monitor a residual register, detect when the current
instance has entered a small-residual regime, and then use the stability estimate to finish the solve
with a prescribed accuracy.

\paragraph{Related work.}
Quantum linear-system algorithms have developed along several complementary directions since
HHL~\cite{HarrowHassidimLloyd2009}. Improvements based on variable-time amplitude amplification,
precision-efficient linear-system solvers, block encodings, and QSVT have led to nearly optimal
black-box QLSA complexity under standard oracle assumptions
\cite{Ambainis2012,ChildsKothariSomma2017,GilyenSuLowWiebe2019}. A separate line of work develops
adiabatic and adiabatic-inspired linear-system solvers, in which the solution state is prepared by
following, randomizing, or discretizing a Hamiltonian path associated with the linear system
\cite{SubasiSommaOrsucci2019,AnLin2022,CostaAnSandersSuBabbushBerry2022,JenningsLostaglioPallisterSornborgerSubasi2025}. These algorithms provide
a conceptually different route to QLSA with near-optimal or optimal dependence on the condition
number and precision, but their stopping times are still prescribed from a priori spectral
information.

Quantum algorithms for elliptic PDEs have also been studied through finite difference, finite
element, and spectral discretizations
\cite{CaoPapageorgiouPetrasTraubKais2013,CladerJacobsSprouse2013,
MontanaroPallister2016,ChildsLiuOstrander2021}. More recent work has emphasized explicit block encodings, boundary
conditions, and quantum-compatible preconditioning
\cite{DeimlPeterseim2024,KharaziFitzpatrickKirbyMullin2025}. Preconditioning is particularly important for PDE linear systems, and multilevel methods such as
the BPX preconditioner provide a natural bridge between classical elliptic theory and quantum
linear algebra~\cite{BramblePasciakXu1990,DeimlPeterseim2024}. Tong, An, Wiebe, and Lin
developed fast inversion as a preconditioning primitive and discussed QLSP costs in terms of
right-hand-side overlap and effective condition numbers~\cite{TongAnWiebeLin2021}. Schr\"odingerization-based approaches have also
been used to connect classical iterative linear-algebra methods and preconditioned linear systems
to quantum simulation frameworks
\cite{JinLiu2024Discrete,JinLiuMaYu2025}.  Recent quantum multigrid work for finite element problems also seeks to bring multilevel
iterative-solver structure into the quantum setting, particularly to address condition-number
growth and the use of initial guesses~\cite{RaisuddinDe2024QuantumMultigrid}. This is complementary
to the present work: qMG introduces a quantum analogue of multigrid iteration, whereas our focus is
on a co-evolving residual register and a residual-based stopping test. These works are closely related in spirit because they
bring numerical linear-algebra structure into quantum algorithms. The present paper differs in that
the residual is embedded into the quantum
dynamics as a measurable register.

Another relevant direction is instance-aware QLSA. Standard complexity estimates are often stated
in terms of matrix-level parameters such as sparsity, normalization, and condition number. In
practice, however, the difficulty of a linear-system instance can depend strongly on the input
state or right-hand side. Li's right-hand-side dependent QLSA~\cite{Li2025BeyondCondition} is a
recent example of this perspective.

\paragraph{Organization.}
The remainder of the paper develops this idea in the elliptic setting. We first review the
Galerkin discretization and the divergence--gradient factorization of the stiffness matrix in \cref{sec:elliptic-galerkin}. We then
explain sparse and block-encoding access to the first-order factor \(G_h\). Next in \cref{sec:ode-solver} we formulate the
first-order ODE solver and derive its stability and complexity estimates.  
We then augment these ODEs with residual variables, dilate the resulting nonunitary dynamics to a
Schr\"odinger equation, and show how measuring a residual register yields a residual-based dynamic
stopping rule in \cref{sec:weak-measurement}.
 Finally, we
present numerical experiments in \cref{sec:numerics} and discuss extensions in \cref{sec:discussion}.

\section{Elliptic discretizations and first-order access}
\label{sec:elliptic-galerkin}

In this section we review the finite element linear systems that will be used throughout the
paper. All objects in this section are classical matrices and vectors. We reserve bra-ket notation,
such as \(\ket{\bm b_h}\) and \(\ket{\bm x_h}\), for the next section, where these vectors are
encoded as normalized quantum states.

Let \(\Omega\subset\R^d\) be a bounded connected Lipschitz domain and consider the elliptic model
problem
\begin{equation}
\label{eq:pde}
-\nabla\cdot(a(x)\nabla u(x))=f(x),
\qquad x\in\Omega,
\qquad
u|_{\partial\Omega}=0 .
\end{equation}
We assume that \(a(x)\in\R^{d\times d}\) is symmetric and uniformly elliptic:
\[
a_{\min}|\bm\psi|^2
\le
\bm\psi^T a(x)\bm\psi
\le
a_{\max}|\bm\psi|^2
\qquad\text{for all }\bm\psi\in\R^d .
\]
The weak formulation is: find \(u(x)\in H_0^1(\Omega)\) such that
\begin{equation}
\label{eq:weak-form}
\mathsf a(u,v)=\mathsf f(v)
\qquad\text{for all } v\in H_0^1(\Omega),
\end{equation}
where from integration by parts, the bilinear form and linear functional are respectively given by,
\[
\mathsf a(u,v)=\int_\Omega \nabla v(x)^T a(x)\nabla u(x)\dd x,
\qquad
\mathsf f(v)=\int_\Omega f(x)v(x)\dd x .
\]
This variational form is the starting point of the Galerkin finite element method \cite{Ciarlet1978,BrennerScott2008}. Instead of
enforcing the PDE pointwise, one asks that the equation hold after testing against all functions in
the energy space.

\subsection{Galerkin discretization}
\label{subsec:galerkin-discretization}

Let \(\mathcal T_h\) be a shape-regular, quasi-uniform triangulation with mesh width \(h\), and let
\(V_h\subset H_0^1(\Omega)\) be a conforming finite element space with basis
\(\{\phi_j\}_{j=1}^{N_h}\). The Galerkin approximation seeks a solution in the subspace spanned by these basis functions, 
\[
u_h(x)=\sum_{j=1}^{N_h} (\bm x_h)_j\phi_j(x)\in V_h
\]
and is determined by forcing the weak form \eqref{eq:weak-form} in the subspace,
\begin{equation}
\label{eq:galerkin}
\mathsf a(u_h,v_h)=\mathsf f(v_h)
\qquad\text{for all }v_h\in V_h .
\end{equation}
Taking \(v_h=\phi_k\) gives the finite-dimensional linear system
\begin{equation}
\label{eq:unscaled-fe-system}
A_h\bm x_h=\bm b_h,
\end{equation}
where
\begin{equation}\label{eq:A_hb_h}
    (A_h)_{kj}
=
\int_\Omega \nabla\phi_k(x)^Ta(x)\nabla\phi_j(x)\dd x,
\qquad
(\bm b_h)_k
=
\int_\Omega f(x)\phi_k(x)\dd x .
\end{equation}
Thus the entries of the matrix and the right-hand side are obtained from local element integrals. Throughout the paper, we will use the grid size $h$ to label the resulting approximations. 

The Galerkin projection leads to a linear system of equations, 
\begin{equation}
\label{eq:Ah-system}
    A_h\bm x_h=\bm b_h,
    \qquad
    A_h=A_h^\dagger\succ0 .
\end{equation}

In the quantum-algorithm analysis below, \(A_h\), \(G_h\), and \(\bm b_h\) are understood in the
discrete \(L^2\)-normalized coordinate system, so that Euclidean vector norms represent the
corresponding discrete \(L^2\) norms. Equivalently, the usual nodal stiffness and load vectors are
transformed by the finite element mass matrix. We keep the notation \(A_h\bm x_h=\bm b_h\) for
the normalized system. In these coordinates,
\[
    A_h=G_h^\dagger G_h,\qquad
    \lambda_{\min}(A_h)=\Theta(1),\qquad
    \lambda_{\max}(A_h)=\Theta(h^{-2}).
\]
For second-order elliptic operators on quasi-uniform meshes, this gives
\(\kappa(A_h)=\Theta(h^{-2})\).
Under the standard
elliptic regularity assumption 
\(u\in H^2(\Omega)\), 
\begin{equation}
\label{eq:L2-error}
\norm{u-u_h}_{L^2(\Omega)}
\le
C h^2\norm{u}_{H^2(\Omega)} .
\end{equation} 
\(f\in L^2(\Omega)\)
is a natural baseline assumption for obtaining the \(\bigO(h^2)\) \(L^2\)-rate ~\cite[Sec. 5.7]{BrennerScott2008}.

\subsection{Divergence--gradient factorization}
\label{subsec:div-grad}

One key structural feature is that the stiffness matrix in \eqref{eq:A_hb_h} admits a Gram factorization
\begin{equation}
\label{eq:Gram}
A_h=G_h^\dagger G_h .
\end{equation}
For finite elements, this factorization can be seen directly from the element quadrature. Suppose
the element integrals are evaluated using quadrature nodes \(x_{Kq}\) and weights \(w_{Kq}\), where
\(K\) indexes elements and \(q\) indexes quadrature points. Then \(G_h\) can be assembled from the gradient of the shape functions evaluated at \(x_{Kq}\) and scaled by  \(w_{Kq}\). In particular \(G_h\) can be regarded as the weighted discrete gradient, and \(G_h^\dagger\) is the
corresponding discrete negative divergence.
For conforming Dirichlet discretizations on quasi-uniform meshes,  one has 
\begin{equation}
\label{eq:Gh-scaling}
\norm{G_h}=\Theta(h^{-1}),
\qquad
\sigma_{\min}(G_h)=\Theta(1).
\end{equation}
We will consider this regime in this paper and highlight the scaling with the grid size $h$. 

 More generally, mimetic finite difference and compatible
discretization methods are designed to preserve discrete gradient--divergence adjointness and
integration-by-parts identities as well
\cite{BrezziLipnikovShashkov2005,HymanShashkovSteinberg1997}, and they can also be considered. For highly structured finite difference or tensor-product discretizations, one may also exploit
arithmetic structure directly in the block encoding, rather than relying only on generic sparse
oracles; see, for example, explicit circuit constructions for structured sparse matrices and
structured data-loading approaches
\cite{CampsLinVanBeeumenYang2024,SunderhaufCampbellCamps2024}.

\subsection{Sparse access and block encodings of \(G_h\)}
\label{subsec:block-encoding-Gh}

We use the standard sparse-access model for matrices. The following lemma is a standard consequence
of sparse block-encoding constructions and matrix arithmetic in the block-encoding/QSVT framework
of Gily\'en et al.~\cite{GilyenSuLowWiebe2019}.

\begin{lemma}[Sparse-access block encoding]
\label{lem:sparse-block-encoding}
Let \(M\in\C^{N_1\times N_2}\) be \(s\)-sparse by rows and columns. Assume access to oracles that,
given a row or column index and an integer \(\ell\in\{1,\ldots,s\}\), return the location and value
of the \(\ell\)-th nonzero entry to precision \(\eps_{\rm val}\). If
\(|M_{ij}|\le M_{\max}\), then one can construct an \((\alpha,a,\eps)\)-block encoding of \(M\)
with
\[
\alpha=\bigO(sM_{\max}),
\qquad
a=\bigO(\log(N_1+N_2)),
\]
using \(\bigO(1)\) calls to the sparse-access oracles and additional gates
\(\polylog((N_1+N_2)/\eps)\).
\end{lemma}

For finite differences and fixed-order finite elements, every row and column of \(G_h\) has only
\(\bigO(1)\) nonzeros (or \(\bigO(d)\) for high-dimensional problems). After mass scaling, the nonzero entries of the discrete gradient are
\(\bigO(h^{-1})\), and \eqref{eq:Gh-scaling} shows that this is also the correct operator-norm
scale. Thus \(G_h\) admits a sparse-access block encoding with normalization
\[
\alpha_G=\bigO(h^{-1}).
\]

The use of \(G_h\) is also consistent with the general Gram-matrix constructions in the
block-encoding and QSVT framework: Given a block encoding of \(G_h/\alpha_G\), one can obtain a
block encoding of
\[
G_h^\dagger G_h/\alpha_G^2
\]
by standard block-encoding composition~\cite{GilyenSuLowWiebe2019}. Thus the abstract construction
of a Gram matrix from an encoded factor is already part of the quantum matrix-arithmetic toolkit.
On the other hand, the ODE method developed below does not require a direct block encoding of the stiffness matrix
\(A_h=G_h^\dagger G_h\).

\begin{corollary}[Block encoding of the first-order gradient operator]
\label{cor:Gh-block-encoding}
Assume fixed spatial dimension, a shape-regular uniform mesh, fixed polynomial degree in the shape function, and
bounded uniformly elliptic coefficients. Then \(G_h\) and \(G_h^\dagger\) admit sparse-access block
encodings with normalization
\[
\alpha_G=\bigO(h^{-1}).
\]
Consequently, the Hermitian first-order matrix
\begin{equation}
\label{eq:first-order-H-G}
H_G=
\begin{bmatrix}
0&-iG_h^\dagger\\
iG_h&0
\end{bmatrix}
=H_G^\dagger
\end{equation}
admits a block encoding with normalization \(\alpha_H=\bigO(h^{-1})\), using \(\bigO(1)\) calls to
block encodings of \(G_h\) and \(G_h^\dagger\). 
\end{corollary}

\begin{remark}
Sparse access to \(G_h\) can be organized locally. Given a basis-function index, one enumerates the
elements touching its support; given an element, one enumerates its local basis functions,
quadrature nodes, and quadrature weights; and the value oracle computes the corresponding entry in
\eqref{eq:A_hb_h}. Thus sparse access is to the local element and quadrature data, not to a
preassembled matrix. The same construction extends to higher-order finite elements when the solution and coefficient
regularity justify them. The constants then depend on the polynomial degree \(p\), but for fixed
\(p\) the mesh dependence of the first-order factor remains \(\bigO(h^{-1})\).
\end{remark}

\subsection{Worst-case estimates based on spectral windows }
\label{subsec:runtime-conservative}

Many QLSAs implicitly use spectral transformations. Let  
\[
    A_h\bm v_j=\lambda_j\bm v_j,
    \qquad
    \bm b_h=\sum_j b_j\bm v_j,
    \qquad
    \bm x_\ast=A_h^{-1}\bm b_h
        =\sum_j \frac{b_j}{\lambda_j}\bm v_j .
\]

 In
HHL-type algorithms \cite{HarrowHassidimLloyd2009}, phase estimation must resolve eigenvalues throughout the promised spectral
range, and the controlled reciprocal map \(\lambda\mapsto 1/\lambda\) must be implemented uniformly
on that range. In QSVT-based QLSA, the same issue appears as a polynomial approximation problem:
one constructs a polynomial that approximates \(1/\lambda\) on a normalized interval such as
\[
    \lambda\in[1/\kappa,1].
\]
The degree of this polynomial is governed by the worst-case endpoint \(1/\kappa\). Thus the
algorithm is designed to succeed uniformly over the promised interval, even if the input vector has
very little weight in the most ill-conditioned modes \cite{TongAnWiebeLin2021}.

It is useful to contrast this with classical solvers, e.g., Krylov or conjugate gradient, which can also be regarded as a
polynomial approximation method: after \(k\) iterations,
\[
    \bm x_k \in \bm x_0+\mathcal K_k(A_h,\bm r_0),
    \qquad
    \mathcal K_k(A_h,\bm r_0)
    =
    \operatorname{span}\{\bm r_0,A_h\bm r_0,\ldots,A_h^{k-1}\bm r_0\},
\]
where \(\bm r_0=\bm b_h-A_h\bm x_0\). These polynomials, however, are not fixed in
advance by a worst-case approximation problem on the full spectral interval. They are generated by
an orthogonalization process starting from the actual residual \(\bm r_0\), and therefore depend on the
spectral measure of \(A_h\) induced by the particular right-hand side and initial guess \cite{Saad2003Iterative}. Just as
importantly, conjugate gradient carries \(\bm r_k\) throughout the computation and uses
\(\norm{\bm r_k}\),  as a built-in stopping test. Thus classical
 methods combine polynomial approximation with an adaptive residual-based stopping rule
\cite{Saad2003Iterative,GolubVanLoan2013Matrix}.

For elliptic PDE inputs, the worst-case design may be conservative, as alluded to in the introduction. If \(f\) is smooth, then the
coefficients \(b_j\) often decay rapidly with the frequency of the eigenmode, thus leading to a much smaller effective spectral window. This naturally leads to the interesting question of whether the residual error can be built into a quantum algorithm.   
The algebraic meaning of the residual is simple. Given an approximate solution \(\bm x_h\), define
\[
    \bm r=\bm b_h-A_h\bm x_h .
\]
If $\bm x_\ast$ and \(\bm e=\bm x_\ast-\bm x_h\) denotes the actual solution and the error, respectively, then
\[
    A_h\bm e_h=\bm r.
\]
Thus the residual is directly tied to the error equation and provides an a posteriori indicator of
convergence.

\section{ODE-based solvers and residual dynamics}
\label{sec:ode-solver}

We now turn the elliptic linear system
\[
    A_h\bm x_\ast=\bm b_h,
    \qquad
    A_h=G_h^\dagger G_h,
\]
into a stable first-order dynamical system. The construction is equation-based: instead of applying
a QLSA directly to the stiffness matrix \(A_h\), we use the divergence--gradient factor \(G_h\) to
build a relaxation dynamics whose limiting solution block is the Galerkin solution.

\subsection{Accumulator via ODE relaxation}
\label{subsec:first-order-ode-no-residual}

The factorization \(A_h=G_h^\dagger G_h\) says that the solution can be characterized by the
first-order relations
\begin{equation}
\label{eq:first-order-constraints}
    \bm q=G_h\bm x,
    \qquad
    G_h^\dagger\bm q=\bm b_h .
\end{equation}
Here \(\bm q\) is the discrete gradient, or flux-like, variable. If one evolves the pair
\((\bm x,\bm q)\) directly and then prepares the normalized joint state, the \(\bm q\)-block may
dominate the norm because \(\bm q_\ast=G_h\bm x_\ast\) and
\(\norm{G_h}=\bigO(h^{-1})\). For state preparation, it is preferable to organize the dynamics so
that \(\bm x(t)\) itself is the accumulator block.

Define the first-order residuals
\begin{equation}
\label{eq:residual-definitions}
    \bm r:=\bm b_h-G_h^\dagger\bm q,
    \qquad
    \bm s:=G_h\bm x-\bm q .
\end{equation}
Here \(\bm r\) can be interpreted as the conservation residual and \(\bm s\) as the flux residual.
We keep \(\bm q\) implicit and evolve
\begin{equation}
\label{eq:accumulator-ode}
\left\{
\begin{aligned}
\dot{\bm x} &= \bm r,\\
\dot{\bm r} &= -G_h^\dagger\bm s,\\
\dot{\bm s} &= G_h\bm r-\bm s,
\end{aligned}
\right.
\qquad
\bm x(0)=0,\quad
\bm r(0)=\bm b_h,\quad
\bm s(0)=0 .
\end{equation}
When needed, the flux variable is recovered from
\begin{equation}
\label{eq:q-implicit}
    \bm q(t):=G_h\bm x(t)-\bm s(t).
\end{equation}

\iffalse 
This follows from a direct differentiation, 
\[
    \dot{\bm q}
    =
    G_h\dot{\bm x}-\dot{\bm s}
    =
    G_h\bm r-(G_h\bm r-\bm s)
    =
    \bm s,
\]
and
\[
    \frac{d}{dt}\bigl(\bm b_h-G_h^\dagger\bm q(t)\bigr)
    =
    -G_h^\dagger\bm s(t)
    =
    \dot{\bm r}(t).
\]
Together with the initial conditions, this gives
\[
    \bm r(t)=\bm b_h-G_h^\dagger\bm q(t),
    \qquad
    \bm s(t)=G_h\bm x(t)-\bm q(t)
\]
for all \(t\). Thus \eqref{eq:accumulator-ode} preserves the same first-order residual identities,
while avoiding the need to store \(\bm q\) as part of the output state.
\fi 

To write the dynamics in a compact form, we let
\[
    \bm w(t)=
    \begin{bmatrix}
    \bm r(t)\\
    \bm s(t)
    \end{bmatrix},
    \qquad
    P_r=\begin{bmatrix}I&0\end{bmatrix},
    \qquad
    M_h=
    \begin{bmatrix}
    0&-G_h^\dagger\\
    G_h&-I
    \end{bmatrix}.
\]
Then
\begin{equation}
\label{eq:x-w-system}
    \dot{\bm x}=P_r\bm w,
    \qquad
    \dot{\bm w}=M_h\bm w .
\end{equation}
Equivalently, with
\begin{equation}
\label{eq:joint-z}
    \bm z(t)=
    \begin{bmatrix}
    \bm x(t)\\
    \bm w(t)
    \end{bmatrix},
    \qquad
    L_h=
    \begin{bmatrix}
    0&P_r\\
    0&M_h
    \end{bmatrix},
\end{equation}
the dynamics can be wriiten as, 
\begin{equation}
\label{eq:joint-system}
    \dot{\bm z}=L_h\bm z .
\end{equation}

 The generator has the first-order mesh scale
\begin{equation}
\label{eq:Lacc-scale}
    \norm{L_h}
    \le
    1+\norm{M_h}
    =
    \bigO(h^{-1}).
\end{equation}
The unbounded mesh dependence enters only through \(G_h\) and \(G_h^\dagger\).

Below we use the standard elliptic finite element setting: the space is conforming, the mesh is
shape-regular and quasi-uniform, and homogeneous Dirichlet boundary conditions are imposed. 

\begin{theorem}[Exponential stability ]
\label{thm:first-order-relaxation}
\label{thm:steady-state}
Let \((\bm x(t),\bm w(t))\) solve \eqref{eq:x-w-system} with
\[
    \bm x(0)=0,
    \qquad
    \bm w(0)=
    \begin{bmatrix}
    \bm b_h\\0
    \end{bmatrix}.
\]
Then \(M_h\) is uniformly exponentially stable: there exist constants
\(C_{\rm st}\ge1\) and \(c_{\rm st}>0\), independent of \(h\), such that
\begin{equation}
\label{eq:semigroup-stability}
    \norm{e^{tM_h}}
    \le
    C_{\rm st}e^{-c_{\rm st}t},
    \qquad t\ge0 .
\end{equation}
Consequently,
\begin{equation}
\label{eq:residual-decay}
    \norm{\bm w(t)}
    \le
    C_{\rm st}e^{-c_{\rm st}t}\norm{\bm w(0)} .
\end{equation}
Moreover, \(\bm x(t)\) converges to a limit \(\bm x_\ast\), and this limit solves
\[
    G_h^\dagger G_h\bm x_\ast=\bm b_h .
\]
In particular, 
\begin{equation}\label{eq:x-tail-bound}
     \norm{\bm x_\ast-\bm x(t)}
    \le
    C_{\rm tail} C_{\rm st}e^{-c_{\rm st}t}\norm{\bm w(0)}, 
\end{equation}
where
\[
    C_{\rm tail}:=\frac{C_{\rm st}}{c_{\rm st}} .
\]
\end{theorem}

The proof is postponed to \cref{appA}.

The theorem holds for any discrete vector \(\bm b_h\). When \(\bm b_h\) is obtained from a function
\(f\), the finite element estimate \eqref{eq:L2-error} describes the discretization error, while
\eqref{eq:x-tail-bound} describes the finite-time relaxation error.

\cref{eq:joint-z} will be referred to as an accumulator, which on one hand encodes the solution of \cref{eq:Ah-system}, and on the other hand, gives a useful output-block estimate. Define, for nonzero
\(\bm x_\ast\),
\begin{equation}\label{Gb}
     \Gamma_b:=
    \frac{\norm{\bm w(0)}}{\norm{\bm x_\ast}}
    =
    \frac{\norm{\bm b_h}}{\norm{\bm x_\ast}} .
\end{equation}
Thus choosing
\[
    T
    \ge
    c_{\rm st}^{-1}
    \log\!\left(
    \frac{4C_{\rm st}(1+C_{\rm tail})\Gamma_b}{\eps}
    \right)
\]
ensures
\[
    \norm{\bm x(T)-\bm x_\ast}
    \le
    \frac{\eps}{4}\norm{\bm x_\ast},
    \qquad
    \norm{\bm w(T)}
    \le
    \frac{\eps}{4}\norm{\bm x_\ast}.
\]
Consequently,
\[
    \norm{\bm x(T)}
    =
    (1+\bigO(\eps))\norm{\bm x_\ast},
\]
and the  solution block $\bm z(t)$ has order-one weight in the joint state:
\begin{equation}
\label{eq:px-constant}
    p_x(T)
    :=
    \frac{\norm{\bm x(T)}^2}
    {\norm{\bm x(T)}^2+\norm{\bm w(T)}^2}
    =
    1-\bigO(\eps^2).
\end{equation}
Thus extracting the \(\bm x\)-block from the accumulator state has constant overhead once the ODE
has relaxed.

\subsection{Quantum implementation and first-order simulation scale}
\label{subsec:dilation}

The accumulator ODE \eqref{eq:joint-system} is the dynamical system implemented by the quantum
ODE primitive. It contains the solution accumulator \(\bm x(t)\) and the decaying residual block
\(\bm w(t)\).   Since the ODE dynamics are non-unitary, the simulation complexity usually depends on the non-Hermitian structure of the generator \cite{JinLiuYu2024,AnChildsLin2023,li2025linear}. For   \eqref{eq:joint-system},  
the generator admits the following decomposition, 
\[
    M_h=-iH_h^{(1)}-\Pi_s,
    \qquad
    H_h^{(1)}
    =
    \begin{bmatrix}
    0&-iG_h^\dagger\\
    iG_h&0
    \end{bmatrix},
    \qquad
    \Pi_s=
    \begin{bmatrix}
    0&0\\
    0&I
    \end{bmatrix}.
\]
In particular, the PDE-dependent Hermitian part has norm
\[
    \norm{H_h^{(1)}}=\norm{G_h}=\bigO(h^{-1}),
\]
while the non-Hermitian damping has norm \(\bigO(1)\). Equivalently, the full accumulator
generator has the scale separation
\[
    L_h=-iH_{\rm big}+K_{\rm big},
\]
where
\[
    H_{\rm big}
    =
    \begin{bmatrix}
    0&0\\
    0&H_h^{(1)}
    \end{bmatrix},
    \qquad
    K_{\rm big}
    =
    \begin{bmatrix}
    0&P_r\\
    0&-\Pi_s
    \end{bmatrix},
\]
with
\[
    \norm{H_{\rm big}}=\bigO(h^{-1}),
    \qquad
    \norm{K_{\rm big}}=\bigO(1).
\]
The fact that the non-Hermitian part has much smaller scaling of $\bigO(1)$ leads to efficient simulation algorithms \cite{JinLiuYu2024,AnChildsLin2023,li2025linear}. 

We now switch to quantum-state notation
\[
    \ket{\bm b_h}:=\frac{\bm b_h}{\norm{\bm b_h}},
    \qquad
    \ket{\bm x_\ast}:=\frac{\bm x_\ast}{\norm{\bm x_\ast}},
\]
and invoke quantum ODE solvers. 

\begin{theorem}[QLSA via a first-order quantum ODE solver]
\label{thm:quantum-first-order-ode}
Assume sparse-access or block-encoding access to \(G_h\) as in
Corollary~\ref{cor:Gh-block-encoding}, and assume that $\bm b_h$ is normalized and \(\ket{\bm b_h}\) can be prepared with
cost \(C_b\). Consider the accumulator ODE \eqref{eq:joint-system} with initial state
\[
    \bm z(0)=
    \begin{bmatrix}
    0\\
    \bm b_h\\
    0
    \end{bmatrix}.
\]
Let \(T\) be chosen so that \(\ket{\bm x(T)}\) is \(\eps\)-close to \(\ket{\bm x_\ast}\) according to \cref{eq:x-tail-bound}. Then an
optimal linear-ODE primitive for nonunitary dynamics, such as the LCHS framework
\cite{AnChildsLin2023}, can prepare an \(\eps\)-accurate normalized solution state with leading
cost
\begin{equation}
\label{eq:qode-accumulator-cost}
    \widetildeO\!\left(
        \frac{\norm{\bm b_h}}{\norm{\bm x_\ast}}
        \left[
        C_b+
        h^{-1}\polylog\frac1\eps
        \right]
    \right),
\end{equation}
up to method-dependent logarithmic factors. If \(\bm b_h\) is normalized before state preparation,
the prefactor is equivalently the standard prefactor in QLSA
\[
    \frac{1}{\norm{A_h^{-1}\ket{\bm b_h}}}.
\]
Moreover, after relaxation the \(\bm x\)-block has probability
\(p_x(T)=1-\bigO(\eps^2)\) in the joint state. Thus extracting the solution block adds only
constant overhead and does not introduce an additional \(h\)-dependent factor.
\end{theorem}

The proof is postponed to \cref{appB}.

\begin{remark}
The prefactor \(\norm{\bm b_h}/\norm{\bm x_\ast}\), or equivalently
\(1/\norm{A_h^{-1}\ket{\bm b_h}}\) for normalized input, is the standard output-normalization
factor for preparing a normalized state after nonunitary evolution. It is not caused by the
first-order auxiliary variables. For elliptic problems, this prefactor has a well defined limit: as the ratio between the $L^2$ norms of $f$ and $u$. 
\end{remark}

This dynamic approach makes another interesting connection between QLSA and quantum ODE algorithms: some ODE algorithms, by using history states, can be reduced to QLSA ~\cite{BerryChildsOstranderWang2017,Krovi2023}. The current algorithm demonstrates that some QLS problems can be solved using quantum ODE solvers.

\section{Residual measurements and dynamic stopping}
\label{sec:weak-measurement}

Recall from the previous section that the accumulator residual--solution dynamics are
\[
    \bm z(t)=
    \begin{bmatrix}
    \bm x(t)\\
    \bm w(t)
    \end{bmatrix},
    \qquad
    \dot{\bm z}=L_h\bm z,
    \qquad
    \bm w=(\bm r,\bm s)^T .
\]
Here \(\bm x(t)\) is the solution accumulator and \(\bm w(t)\) is the mixed residual block. The
purpose of this section is to use the \(\bm w\)-block as a measurable stopping signal. Instead of
choosing the runtime only from a worst-case spectral estimate, we estimate the residual-block
amplitude and use the stability estimate to decide when the current instance has entered the final
exponentially decaying regime.

\subsection{Residual probability from the joint state}

We use bra-ket notation for the amplitude-encoded blocks of the joint vector. Thus
\(\ket{\bm x(t)}\) and \(\ket{\bm w(t)}\) denote the computational-basis amplitude vectors
associated with \(\bm x(t)\) and \(\bm w(t)\), with the overall normalization written explicitly.
The dilated joint dynamics prepares, up to implementation error, a state of the form
\begin{equation}
\label{eq:unscaled-extended-state}
\ket{\psi(t)}
=
\frac{
\ket{0}\ket{\bm x(t)}+\ket{1}\ket{\bm w(t)}
}{
\sqrt{\norm{\bm x(t)}^2+\norm{\bm w(t)}^2}
}.
\end{equation}
Let
\[
    \Pi_{\res}=\ket{1}\bra{1}\otimes I
\]
be the projector onto the residual block. Then
\begin{equation}
\label{eq:pres}
p_{\res}(t)
:=
\bra{\psi(t)}\Pi_{\res}\ket{\psi(t)}
=
\frac{\norm{\bm w(t)}^2}
{\norm{\bm x(t)}^2+\norm{\bm w(t)}^2}.
\end{equation}

The connection between this probability and convergence is the tail estimate from
\cref{thm:steady-state}:
\begin{equation}
\label{eq:x-w-certificate-bound}
\norm{\bm x_\ast-\bm x(t)}
\le
C_{\rm tail}\norm{\bm w(t)},
\qquad
C_{\rm tail}:=\frac{C_{\rm st}}{c_{\rm st}} .
\end{equation}
Thus the residual block is not merely diagnostic; it controls the remaining distance to the steady
state.

\begin{lemma}[Residual probability controls the relative tail]
\label{lem:residual-probability-tail}
Assume \(\norm{\bm x(t)}>0\). Fix $p \in (0,1)$. If
\[
    p_{\res}(t)\le p,
\]
then
\[
    \frac{\norm{\bm x_\ast-\bm x(t)}}{\norm{\bm x(t)}}
    \le
    C_{\rm tail}
    \sqrt{\frac{p}{1-p}} .
\]
\end{lemma}

\begin{proof}
From \eqref{eq:pres},
\[
    \frac{\norm{\bm w(t)}^2}{\norm{\bm x(t)}^2}
    =
    \frac{p_{\res}(t)}{1-p_{\res}(t)} .
\]
Combining this identity with \eqref{eq:x-w-certificate-bound} gives the claim.
\end{proof}

\subsection{Constant-threshold readout and deterministic coasting}
\label{subsec:constant-threshold-coasting}

A direct stopping rule based on the final tolerance would require
\[
    p_{\res}(t)
    \lesssim
    \frac{\eps^2}{C_{\rm tail}^2+\eps^2},
\]
which is an \(\mathcal O(\eps^2)\)-probability event. Estimating such a small probability by
direct sampling would require \(\mathcal O(\eps^{-2})\) shots. We avoid this by using the
exponential decay of the residual block.

Choose a fixed constant threshold \(p_0\in(0,1)\), independent of \(\eps\), and define
\[
    \eta_0:=\sqrt{\frac{p_0}{1-p_0}},
    \qquad
    \beta_0:=C_{\rm tail}\eta_0 .
\]
We choose \(p_0\) so that
\[
    \beta_0<1 .
\]

Suppose that a checkpoint time \(t_k\) satisfies
\[
    p_{\res}(t_k)\le p_0 .
\]
Then
\[
    \norm{\bm w(t_k)}
    \le
    \eta_0\norm{\bm x(t_k)}
\]
and, by Lemma \ref{lem:residual-probability-tail},
\[
    \norm{\bm x_\ast-\bm x(t_k)}
    \le
    \beta_0\norm{\bm x(t_k)} .
\]
Since \(\beta_0<1\), the current accumulator norm is comparable to the solution norm:
\[
    \norm{\bm x(t_k)}
    \le
    \frac{1}{1-\beta_0}\norm{\bm x_\ast}.
\]
Therefore,
\[
    \norm{\bm w(t_k)}
    \le
    \frac{\eta_0}{1-\beta_0}\norm{\bm x_\ast}.
\]

Now continue the same ODE for an additional deterministic time \(\Delta t\). Set
\[
    T_\star=t_k+\Delta t.
\]
The residual decay gives
\[
    \norm{\bm w(T_\star)}
    \le
    C_{\rm st}e^{-c_{\rm st}\Delta t}\norm{\bm w(t_k)} .
\]
Combining this with the tail estimate gives
\[
    \norm{\bm x_\ast-\bm x(T_\star)}
    \le
    \frac{C_{\rm tail}C_{\rm st}\eta_0}{1-\beta_0}
    e^{-c_{\rm st}\Delta t}
    \norm{\bm x_\ast}.
\]
Thus it is sufficient to choose
\begin{equation}
\label{eq:coasting-time}
    \Delta t
    \ge
    c_{\rm st}^{-1}
    \log\!\left(
    \frac{C_{\rm tail}C_{\rm st}\eta_0}{(1-\beta_0)\eps}
    \right)
\end{equation}
to ensure the relative error reaches the given precision
\[
    \norm{\bm x_\ast-\bm x(T_\star)}
    \le
    \eps\norm{\bm x_\ast}.
\]

This calculation suggests the stopping strategy used below: the checkpoint measurements only test for entry into
a constant-residual regime, while the final accuracy is supplied by a deterministic coasting time $\Delta t$
computed from the stability estimate. Hence the sampling problem is a constant-probability
readout, rather than an \(\eps^2\)-probability readout.

\begin{theorem}[Constant-threshold dynamic stopping]
\label{thm:constant-threshold-stopping}
Let \(p_0\in(0,1)\) satisfy
\[
    C_{\rm tail}\sqrt{\frac{p_0}{1-p_0}}<1 .
\]
Suppose that \(t_{\rm ent}\) is the first checkpoint time at which
\begin{equation}\label{p<p0}
     p_{\res}(t_{\rm ent})\le p_0 .
\end{equation}

Let \(T_\star=t_k+\Delta t\), with \(\Delta t\) chosen as in
\eqref{eq:coasting-time}. Then
\[
    \norm{\bm x_\ast-\bm x(T_\star)}
    \le
    \eps\norm{\bm x_\ast}.
\]
Moreover, the final joint state
\[
    \begin{bmatrix}
    \bm x(T_\star)\\
    \bm w(T_\star)
    \end{bmatrix}
\]
has solution-block probability \(1-\mathcal O(\eps^2)\).
\end{theorem}

\subsection{Checkpoint sampling}
\label{subsec:checkpoint-sampling}
Now we explain how to detect the check point that reaches the threshhold  \cref{p<p0}.

Fix a checkpoint schedule
\[
    0<t_0<t_1<\cdots .
\]
For each candidate time \(t_k\), a checkpoint circuit prepares a fresh copy of the initial state,
evolves the dilated ODE to time \(t_k\), measures the residual flag, and then discards the state.
Thus one checkpoint circuit at time \(t_k\) produces one Bernoulli sample with success probability
\[
    p_k:=p_{\res}(t_k).
\]
If \(N_{\rm shot}\) independent checkpoint circuits are run at the same time \(t_k\), then the outcome 
\[
    X_k\sim {\rm Binomial}(N_{\rm shot},p_k),
    \qquad
    \widehat p_k=\frac{X_k}{N_{\rm shot}} .
\]

The stopping-test stage examines the checkpoint times sequentially. At time \(t_k\), it accepts
that checkpoint if the estimated residual probability is below the constant threshold \(p_0\). A
simple rule is
\[
    \widehat p_k\le p_0.
\]
More conservatively, one may use a confidence buffer $\delta$ and accept only when
\begin{equation}
\label{eq:buffered-threshold}
    \widehat p_k+\delta \le p_0,
    \qquad
    \delta=\sqrt{\frac{\log(2K/\nu)}{2N_{\rm shot}}},
\end{equation}
where \(K\) is the number of checkpoint times tested and \(\nu\) is an acceptable failure probability.
The first index satisfying this buffered test is denoted
\[
    \widehat k_{\rm ent},
    \qquad
    \widehat t_{\rm ent}:=t_{\widehat k_{\rm ent}} .
\]
The final production time is then
\[
    T_\star=\widehat t_{\rm ent}+\Delta t,
\]
where \(\Delta t\) is the deterministic coasting time from \eqref{eq:coasting-time}.

By Hoeffding's inequality and a union bound over the \(K\) checkpoint times, the buffered rule
controls the probability of falsely accepting a checkpoint whose true residual probability is above
\(p_0\) by at most \(\nu\). For a fixed confidence buffer, this gives
\[
    N_{\rm shot}
    =
    \mathcal O(\log(K/\nu)).
\]
Thus the stopping-test sampling overhead is logarithmic in the number of checkpoints and the
failure-probability parameter, and it is not polynomial in \(1/\eps\).

For a geometric checkpoint schedule \(t_k=2^k t_0\), the number of tested checkpoints is logarithmic
in the accepted entry time \(t_{\rm ent}\). Since the final coasting time satisfies
\(\Delta t=\mathcal O(\log(1/\eps))\), the total evolution time retains the same logarithmic
precision dependence as the a priori stability bound, while the measured entry time remains
instance dependent.

\subsection{Algorithmic summary}
\label{subsec:algorithmic-summary}

The residual-based solver separates the stopping-test stage from the final production run. The
stopping-test stage scans a sequence of checkpoint times \(t_0,t_1,\ldots\). At each candidate time
\(t_k\), it performs \(N_{\rm shot}\) independent checkpoint runs. A single checkpoint run means:
prepare the initial joint residual--solution state, evolve the dilated ODE once to time \(t_k\),
measure the residual flag, record one binary outcome, and discard the state. Thus the residual
measurement is not used to continue the same quantum trajectory; it is only a Bernoulli sample for
estimating \(p_{\res}(t_k)\).

Once the first accepted checkpoint time \(\widehat t_{\rm ent}\) is selected, the final production
run evolves a fresh initial state directly to
\[
    T_\star=\widehat t_{\rm ent}+\Delta t,
\]
where \(\Delta t\) is the deterministic coasting time in \eqref{eq:coasting-time}. The production
run is not interrupted by residual measurements.

\begin{algorithm}[htp]
\caption{Residual-based quantum elliptic solver with dynamic stopping}
\label{alg:residual-based}
\begin{algorithmic}[1]
\Require Sparse-access or block-encoding access to \(G_h\); state preparation for
\(\ket{\bm b_h}\); target precision \(\eps\); constant checkpoint threshold \(p_0\); checkpoint
times \(t_0<t_1<t_2<\cdots\); number of checkpoint samples \(N_{\rm shot}\).
\Ensure A normalized solution state \(\ket{\widetilde{\bm x}_h}\) satisfying the stopping rule in
\cref{thm:constant-threshold-stopping}.
\State Construct the joint accumulator ODE
\[
    \dot{\bm z}=L_h\bm z,
    \qquad
    \bm z=(\bm x,\bm w)^T,
\]
and its dilation to a Schr\"odinger equation.
\For{\(k=0,1,2,\ldots\)}
  \State Run \(N_{\rm shot}\) independent checkpoint circuits:
  prepare the initial joint state, evolve once to time \(t_k\), and measure the residual flag.
  \State Let \(\widehat p_{\res}(t_k)\) be the empirical frequency of the residual outcome.
  \If{\(\widehat p_{\res}(t_k)+\delta \le p_0\), with \(\delta \) as in
  \eqref{eq:buffered-threshold}}
    \State Set \(t_{\rm ent}\gets t_k\) and exit the loop.
  \EndIf
\EndFor
\State Set \(T_\star\gets t_{\rm ent}+\Delta t\), with \(\Delta t\) chosen by
\eqref{eq:coasting-time}.
\State Prepare a fresh initial state and evolve the dilated joint system once to time \(T_\star\),
without residual measurements during the evolution.
\State Output the solution block, namely the \(\bm x\)-subregister of \(\bm z(T_\star)\), with the
appropriate normalization and postselection factors.
\end{algorithmic}
\end{algorithm}

We summarize the algorithm in \cref{alg:residual-based}.
 The checkpoint circuits are used only to
detect when the current instance has entered a constant-residual regime. The final accuracy is then
obtained by deterministic coasting under the proven exponential stability of the residual block.
The final solution state is produced in a separate run.

\begin{figure}[thp]
\centering
\begin{quantikz}[row sep={0.75cm,between origins}, column sep=0.35cm]
\lstick{$a_{\res}:\ket{0}$}
  & \qw
  & \gate[wires=2]{U_{\mathrm{dil}}(t_k)}
  & \meter{}
  & \gate{\mathrm{reset}}
  & \qw
  & \rstick{$m_k\in\{0,1\}$}
\\
\lstick{$d:\ket{\psi_0}$}
  & \gate{U_{\mathrm{prep}}}
  & \qw
  & \qw
  & \qw
  & \qw
  & \qw
\\
\lstick{$c$}
  & \cw
  & \cw
  & \gate[1][3.2cm]{\begin{array}{c}
      \text{accumulate samples}\\[-1mm]
      \scriptsize \widehat p_{\res}(t_k)=X_k/N_{\rm shot}
    \end{array}}
  & \gate[1][3.3cm]{\begin{array}{c}
      \text{entry test}\\[-1mm]
      \scriptsize \widehat p_{\res}(t_k)+\delta\le p_0
    \end{array}}
  & \cw
  & \cw
\arrow[from=1-4,to=3-4]
\arrow[from=3-5,to=1-5]
\end{quantikz}
\caption{Checkpoint circuit for residual-based stopping. Each checkpoint run prepares the initial
state, evolves the dilated joint accumulator system once to time \(t_k\), and measures the residual
flag \(a_{\res}\). The outcome \(m_k\in\{0,1\}\) is one Bernoulli sample. Repeating the checkpoint
circuit \(N_{\rm shot}\) times gives the estimate
\(\widehat p_{\res}(t_k)=X_k/N_{\rm shot}\), which is compared with the constant threshold
\(p_0\). If the entry test accepts, a separate production run evolves directly to
\(T_\star=t_k+\Delta t\), without residual measurements, and outputs the solution block.}
\label{fig:dynamic-circuit-certification}
\end{figure}
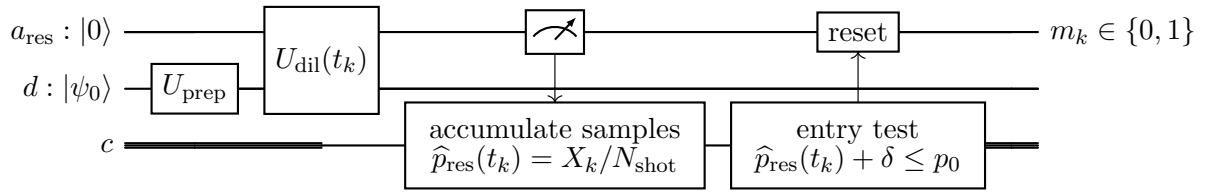

\Cref{fig:dynamic-circuit-certification} shows a possible workflow on hardware that supports
mid-circuit measurement. The reset prepares the residual ancilla
for the next checkpoint shot; it is not a refresh of an ongoing time evolution. Since each
checkpoint run evolves directly from the initial state to \(t_k\), the dilation only needs one time
segment per shot.
\section{Numerical experiments}
\label{sec:numerics}

We illustrate the residual-based stopping mechanism on a standard finite element discretization of
the Poisson equation on the unit square ($a(x) = I$),
\begin{equation}
\label{eq:numerical-poisson}
    -\Delta u=f
    \quad \text{in } \Omega=(0,1)^2,
    \qquad
    u=0
    \quad \text{on } \partial\Omega .
\end{equation}
The domain is partitioned into \(n\times n\) square cells, and each square is split along the same
diagonal to obtain a conforming triangular mesh. We use continuous piecewise-linear, or \(P_1\),
finite elements, with homogeneous Dirichlet boundary conditions imposed by eliminating the boundary
degrees of freedom. In the experiments below, \(n=16\), so \(h=1/16\). The mesh contains
\(16^2=256\) square cells and \(512\) triangular elements. After imposing the boundary condition,
there are \((16-1)^2=225\) interior degrees of freedom. The discrete gradient has two components on
each triangle, giving \(2\cdot512=1024\) gradient degrees of freedom.

The test problem is chosen by the method of manufactured solutions. We take
\begin{equation}
\label{eq:numerical-exact-solution}
u_{\rm ex}(x,y)
=
\sin(\pi x)\sin(\pi y)
+\frac{1}{2}\sin(3\pi x)\sin(2\pi y)
+\frac{1}{4}\sin(5\pi x)\sin(4\pi y),
\end{equation}
and define \(f=-\Delta u_{\rm ex}\). Thus
\begin{equation}
\label{eq:numerical-rhs}
\begin{aligned}
f(x,y)
&=
2\pi^2\sin(\pi x)\sin(\pi y)
+\frac{13\pi^2}{2}\sin(3\pi x)\sin(2\pi y) 
+\frac{41\pi^2}{4}\sin(5\pi x)\sin(4\pi y).
\end{aligned}
\end{equation}
This gives a smooth right-hand side with several spectral components, while still allowing
comparison with a known exact solution. The mesh and the corresponding finite element solution are
shown in \cref{fig:fem-mesh-solution}.

\begin{figure}[bhpt]
\centering
\includegraphics[width=0.9\linewidth]{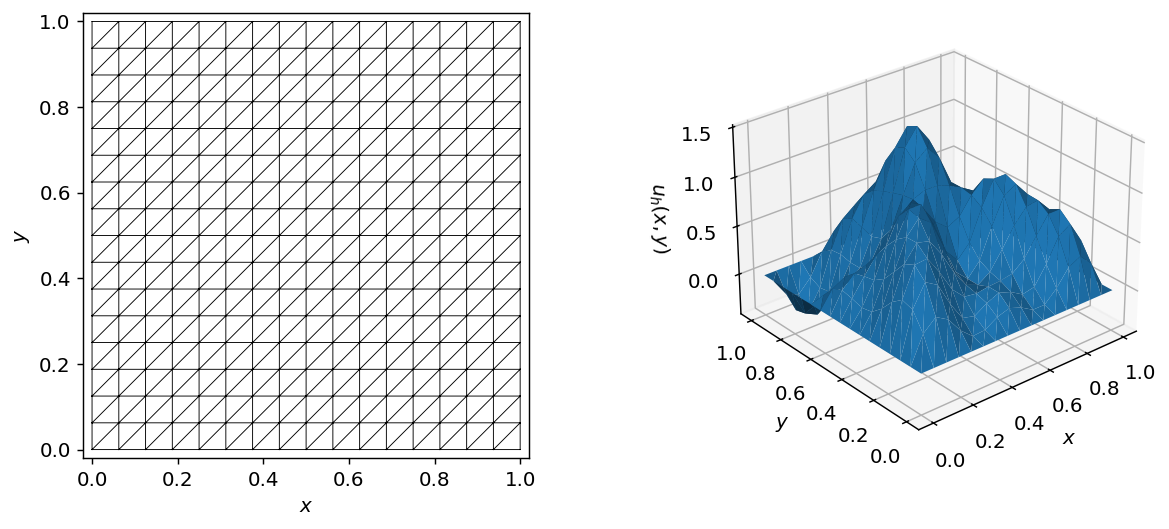}
\caption{Finite element setup for the model Poisson problem. Left: the uniform triangular mesh
obtained by splitting each square cell into two triangles. Right: the \(P_1\) finite element
solution \(u_h\) for the manufactured solution in \eqref{eq:numerical-exact-solution}. The
factorization \(A_h=G_h^\dag G_h\) is assembled from the elementwise gradients of the \(P_1\) basis
functions on this mesh.}
\label{fig:fem-mesh-solution}
\end{figure}

We next evolve the residual-augmented first-order system in \cref{eq:accumulator-ode}.
The flux variable is recovered as
\[
    \bm q(t)=G_h\bm x(t)-\bm s(t).
\]
Thus
\[
    \bm r(t)=\bm b_h-G_h^T\bm q(t),
    \qquad
    \bm s(t)=G_h\bm x(t)-\bm q(t),
\]
so \((\bm r,\bm s)\) is the mixed first-order residual carried by the ODE. The algebraic residual is
\[
    \bm r_A(t):=\bm b_h-A_h\bm x(t)
    =
    \bm r(t)-G_h^\dag\bm s(t).
\]
For the residual register, we use
\(
    \bm w=(\bm r,\bm s),
\)
and measure the probability
\begin{equation}
\label{eq:numerical-pres}
    p_{\res}(t)
    =
    \frac{\norm{\bm w(t)}^2}
    {\norm{\bm x(t)}^2+\norm{\bm w(t)}^2}.
\end{equation}

\begin{figure}[t]
\centering
\begin{minipage}{0.49\linewidth}
\centering
\includegraphics[width=\linewidth]{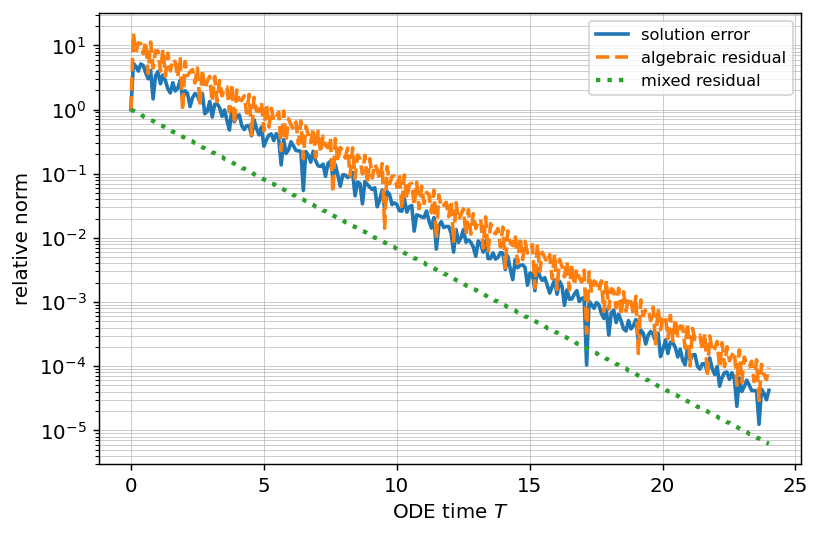}
\end{minipage}
\hfill
\begin{minipage}{0.49\linewidth}
\centering
\includegraphics[width=\linewidth]{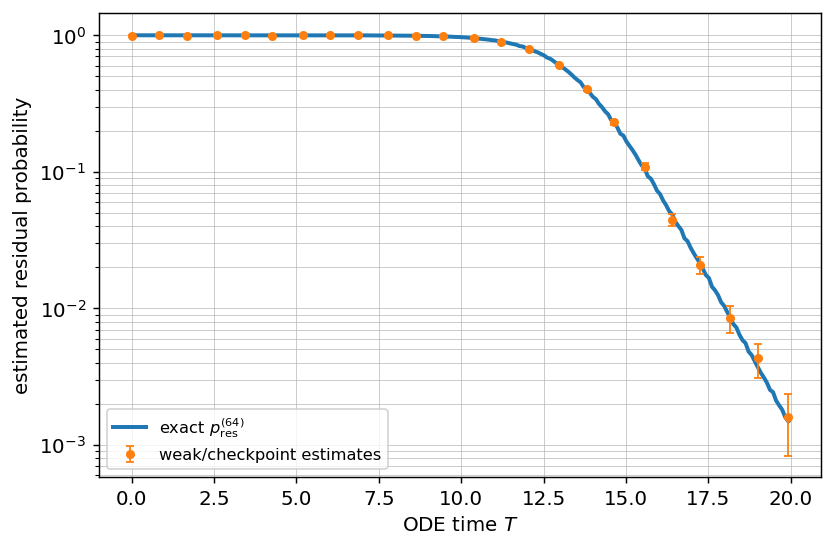}
\end{minipage}
\caption{Residual dynamics and checkpoint readout. Left: relative solution error, relative
algebraic residual \(\norm{\bm r_A(t)}/\norm{\bm b_h}\), mixed residual
\(\norm{\bm w(t)}/\norm{\bm w(0)}\), and accumulator residual probability
\(p_{\res}(t)\), computed from \((\bm x(t),\bm w(t))\). Right:
checkpoint estimates of \(p_{\res}(t)\) from Bernoulli samples of the residual flag.}
\label{fig:numerical-residual-measurement}
\end{figure}

\Cref{fig:numerical-residual-measurement} shows the main numerical mechanism. The algebraic
residual and the mixed residual decay together with the solution error, up to the oscillations
expected from the first-order damped dynamics. The residual-register probability
\(p_{\res}(t)\) follows the same trend. Thus measuring the residual register gives an
observable proxy for convergence, without reconstructing \(\bm x(t)\) or evaluating
\(\bm b_h-A_h\bm x(t)\) classically. The right panel shows that this probability can be estimated
directly from repeated checkpoint measurements.

\section{Summary and Discussions }
\label{sec:discussion}

This paper develops a residual-based dynamic stopping mechanism for quantum linear-system
algorithms, with elliptic PDEs as an important motivating application. The central idea is to make the residual a
co-evolving part of the quantum dynamics rather than a quantity reconstructed after the solution
state has been prepared. For elliptic discretizations with \(A_h=G_h^\dagger G_h\), the first-order
factor \(G_h\) gives the mesh-dependent scale \(\norm{G_h}=\bigO(h^{-1})\), while the residual
variables are driven by the same stable first-order dynamics. After dilation to a Schr\"odinger
evolution, the residual block can be measured in checkpoint runs, producing an instance-dependent
stopping signal. Thus the contribution is not a new worst-case QLSA condition-number improvement,
but a mechanism for making QLSA-based PDE solvers behave more like adaptive algorithms: run the
dynamics, monitor a residual, and stop when the current instance has converged. Although the construction is formulated for elliptic PDEs, the residual-based stopping idea can
also be extended to general stable linear systems \(A\bm x=\bm b\), by using an ODE relaxation in
which the algebraic residual evolves under the same stable dynamics. The elliptic setting is special
because the factorization \(A_h=G_h^\dagger G_h\) exposes the first-order scale
\(\norm{G_h}=\bigO(h^{-1})\), which a generic stable system need not provide.

\smallskip 

Several implementation questions are interesting. Since the residual register is a stopping mechanism rather
than a replacement for existing QLSA primitives, a complete resource estimate should separate the
cost of the underlying nonunitary ODE solver, the output-normalization factor, the checkpoint
sampling overhead, and the final solution readout. The accumulator formulation is useful in this
respect: after relaxation the solution block has order-one weight, so extracting \(\bm x(T)\) does
not introduce an additional mesh-dependent block-selection penalty. At the same time, estimating
absolute quantities may still require estimating \(\norm{\bm x_\ast}\), which can be obtained from
the output-normalization or postselection amplitude of the ODE primitive. It may also be useful to
combine the residual flag with the dilation ancillas for ODE algorithms \cite{JinLiuYu2024,AnLiuLin2023,li2025linear}, so that checkpoint shots estimate both the
conditional residual probability and the simulation success probability. Finally, the constants in
the stability estimate enter only when converting a measured residual signal into a rigorous error
bound; conservative constants give a valid but possibly cautious coasting time, while the measured
entry time into the small-residual regime remains instance dependent.

\smallskip 

 For many PDE
applications, the desired output is not the whole solution state but a quantity of interest, such
as an energy, boundary flux, averaged field value, or response functional. This suggests developing
observable-specific stopping rules, possibly using dual-weighted or adjoint residuals. Another
natural direction is nonsymmetric steady-state problems such as advection--diffusion, where the
diffusion part still has a coercive \(G_h^\dagger G_h\) structure while the transport term acts as a
lower-order or skew perturbation. Understanding when such systems admit stable first-order
dilations with measurable residual registers would make residual-based dynamic stopping applicable
to a broader class of quantum PDE algorithms.

\section*{Acknowledgments}
This work was supported by NSF Grant DMS-2411120.

\bibliographystyle{plain}
\bibliography{qlsa}

@article{HarrowHassidimLloyd2009,
  author        = {Harrow, Aram W. and Hassidim, Avinatan and Lloyd, Seth},
  title         = {{Quantum Algorithm for Linear Systems of Equations}},
  journal       = {Physical Review Letters},
  volume        = {103},
  pages         = {150502},
  year          = {2009},
  doi           = {10.1103/PhysRevLett.103.150502},
  eprint        = {0811.3171},
  archivePrefix = {arXiv},
  primaryClass  = {quant-ph}
}

@misc{RaisuddinDe2024QuantumMultigrid,
  author        = {Raisuddin, Osama Muhammad and De, Suvranu},
  title         = {{Quantum Multigrid Algorithm for Finite Element Problems}},
  year          = {2024},
  eprint        = {2404.07466},
  archivePrefix = {arXiv},
  primaryClass  = {quant-ph},
  doi           = {10.48550/arXiv.2404.07466}
}

@article{JenningsLostaglioPallisterSornborgerSubasi2025,
  author        = {Jennings, David and Lostaglio, Matteo and Pallister, Sam and Sornborger, Andrew T. and Suba{\c{s}}{\i}, Yi{\u{g}}it},
  title         = {{Randomized Adiabatic Quantum Linear Solver Algorithm with Optimal Complexity Scaling and Detailed Running Costs}},
  journal       = {{PRX Quantum}},
  volume        = {6},
  number        = {4},
  pages         = {040373},
  year          = {2025},
  doi           = {10.1103/1xkb-22cc},
  eprint        = {2305.11352},
  archivePrefix = {arXiv},
  primaryClass  = {quant-ph}
}

@article{JenningsLostaglioLowriePallisterSornborger2024,
  author        = {Jennings, David and Lostaglio, Matteo and Lowrie, Robert B. and Pallister, Sam and Sornborger, Andrew T.},
  title         = {{The Cost of Solving Linear Differential Equations on a Quantum Computer: Fast-Forwarding to Explicit Resource Counts}},
  journal       = {{Quantum}},
  volume        = {8},
  pages         = {1553},
  year          = {2024},
  doi           = {10.22331/q-2024-12-10-1553},
  eprint        = {2309.07881},
  archivePrefix = {arXiv},
  primaryClass  = {quant-ph}
}

@inproceedings{Ambainis2012,
  author        = {Ambainis, Andris},
  title         = {{Variable Time Amplitude Amplification and Quantum Algorithms for Linear Algebra Problems}},
  booktitle     = {29th International Symposium on Theoretical Aspects of Computer Science (STACS 2012)},
  series        = {Leibniz International Proceedings in Informatics},
  volume        = {14},
  pages         = {636--647},
  year          = {2012},
  doi           = {10.4230/LIPIcs.STACS.2012.636},
  eprint        = {1010.4458},
  archivePrefix = {arXiv},
  primaryClass  = {quant-ph}
}

@article{ChildsKothariSomma2017,
  author        = {Childs, Andrew M. and Kothari, Robin and Somma, Rolando D.},
  title         = {{Quantum Algorithm for Systems of Linear Equations with Exponentially Improved Dependence on Precision}},
  journal       = {SIAM Journal on Computing},
  volume        = {46},
  number        = {6},
  pages         = {1920--1950},
  year          = {2017},
  doi           = {10.1137/16M1087072},
  eprint        = {1511.02306},
  archivePrefix = {arXiv},
  primaryClass  = {quant-ph}
}

@inproceedings{GilyenSuLowWiebe2019,
  author        = {Gily{\'e}n, Andr{\'a}s and Su, Yuan and Low, Guang Hao and Wiebe, Nathan},
  title         = {{Quantum Singular Value Transformation and Beyond: Exponential Improvements for Quantum Matrix Arithmetics}},
  booktitle     = {Proceedings of the 51st Annual ACM SIGACT Symposium on Theory of Computing},
  pages         = {193--204},
  year          = {2019},
  doi           = {10.1145/3313276.3316366},
  eprint        = {1806.01838},
  archivePrefix = {arXiv},
  primaryClass  = {quant-ph}
}

@article{SubasiSommaOrsucci2019,
  author        = {Suba{\c{s}}{\i}, Yi{\u{g}}it and Somma, Rolando D. and Orsucci, Davide},
  title         = {{Quantum Algorithms for Systems of Linear Equations Inspired by Adiabatic Quantum Computing}},
  journal       = {Physical Review Letters},
  volume        = {122},
  pages         = {060504},
  year          = {2019},
  doi           = {10.1103/PhysRevLett.122.060504},
  eprint        = {1805.10549},
  archivePrefix = {arXiv},
  primaryClass  = {quant-ph}
}

@article{AnLin2022,
  author        = {An, Dong and Lin, Lin},
  title         = {{Quantum Linear System Solver Based on Time-Optimal Adiabatic Quantum Computing and Quantum Approximate Optimization Algorithm}},
  journal       = {ACM Transactions on Quantum Computing},
  volume        = {3},
  number        = {2},
  pages         = {5:1--5:28},
  year          = {2022},
  doi           = {10.1145/3498331},
  eprint        = {1909.05500},
  archivePrefix = {arXiv},
  primaryClass  = {quant-ph}
}

@book{Varga2000MatrixIterative,
  author    = {Varga, Richard S.},
  title     = {{Matrix Iterative Analysis}},
  edition   = {2},
  publisher = {Springer},
  address   = {Berlin},
  year      = {2000},
  doi       = {10.1007/978-3-642-05156-2}
}

@article{CostaAnSandersSuBabbushBerry2022,
  author        = {Costa, Pedro C. S. and An, Dong and Sanders, Yuval R. and Su, Yuan and Babbush, Ryan and Berry, Dominic W.},
  title         = {{Optimal Scaling Quantum Linear-Systems Solver via Discrete Adiabatic Theorem}},
  journal       = {PRX Quantum},
  volume        = {3},
  pages         = {040303},
  year          = {2022},
  doi           = {10.1103/PRXQuantum.3.040303},
  eprint        = {2111.08152},
  archivePrefix = {arXiv},
  primaryClass  = {quant-ph}
}

@article{CladerJacobsSprouse2013,
  author        = {Clader, B. David and Jacobs, Bryan C. and Sprouse, Chad R.},
  title         = {{Preconditioned Quantum Linear System Algorithm}},
  journal       = {Physical Review Letters},
  volume        = {110},
  pages         = {250504},
  year          = {2013},
  doi           = {10.1103/PhysRevLett.110.250504},
  eprint        = {1301.2340},
  archivePrefix = {arXiv},
  primaryClass  = {quant-ph}
}

@article{MontanaroPallister2016,
  author        = {Montanaro, Ashley and Pallister, Sam},
  title         = {{Quantum Algorithms and the Finite Element Method}},
  journal       = {Physical Review A},
  volume        = {93},
  pages         = {032324},
  year          = {2016},
  doi           = {10.1103/PhysRevA.93.032324},
  eprint        = {1512.05903},
  archivePrefix = {arXiv},
  primaryClass  = {quant-ph}
}

@article{OrsucciDunjko2021,
  author        = {Orsucci, Davide and Dunjko, Vedran},
  title         = {{On Solving Classes of Positive-Definite Quantum Linear Systems with Quadratically Improved Runtime in the Condition Number}},
  journal       = {Quantum},
  volume        = {5},
  pages         = {573},
  year          = {2021},
  doi           = {10.22331/q-2021-11-08-573},
  eprint        = {2101.11868},
  archivePrefix = {arXiv},
  primaryClass  = {quant-ph}
}

@misc{li2025linear,
  author = {Xiantao Li},
  title = {From Linear Differential Equations to Unitaries: A Moment-Matching Dilation Framework with Near-Optimal Quantum Algorithms},
  year = {2025},
  eprint = {2507.10285},
  archivePrefix = {arXiv}
}

@article{CaoPapageorgiouPetrasTraubKais2013,
  author        = {Cao, Yudong and Papageorgiou, Anargyros and Petras, Iasonas and Traub, Joseph and Kais, Sabre},
  title         = {{Quantum Algorithm and Circuit Design Solving the Poisson Equation}},
  journal       = {New Journal of Physics},
  volume        = {15},
  pages         = {013021},
  year          = {2013},
  doi           = {10.1088/1367-2630/15/1/013021},
  eprint        = {1207.2485},
  archivePrefix = {arXiv},
  primaryClass  = {quant-ph}
}

@article{BramblePasciakXu1990,
  author        = {Bramble, James H. and Pasciak, Joseph E. and Xu, Jinchao},
  title         = {{Parallel Multilevel Preconditioners}},
  journal       = {Mathematics of Computation},
  volume        = {55},
  number        = {191},
  pages         = {1--22},
  year          = {1990},
  doi           = {10.1090/S0025-5718-1990-1023042-6}
}

@misc{DeimlPeterseim2024,
  author        = {Deiml, Matthias and Peterseim, Daniel},
  title         = {{Quantum Realization of the Finite Element Method}},
  year          = {2024},
  eprint        = {2403.19512},
  archivePrefix = {arXiv},
  primaryClass  = {quant-ph}
}

@article{KharaziFitzpatrickKirbyMullin2025,
  author        = {Kharazi, Tyler and Fitzpatrick, Aaron and Kirby, William M. and Mullin, William J.},
  title         = {{Explicit Block Encodings of Boundary Value Problems for Many-Body Elliptic Operators}},
  journal       = {Quantum},
  volume        = {9},
  pages         = {1764},
  year          = {2025},
  doi           = {10.22331/q-2025-06-04-1764},
  eprint        = {2407.18347},
  archivePrefix = {arXiv},
  primaryClass  = {quant-ph}
}

@misc{JinLiuMaYu2025,
  author        = {Jin, Shi and Liu, Nana and Ma, Yue and Yu, Yue},
  title         = {{Quantum Preconditioning Method for Linear Systems Problems}},
  year          = {2025},
  eprint        = {2505.06866},
  archivePrefix = {arXiv},
  primaryClass  = {quant-ph}
}

@article{JinLiu2024Discrete,
  author        = {Jin, Shi and Liu, Nana},
  title         = {{Quantum Simulation of Discrete Linear Dynamical Systems and Simple Iterative Methods in Linear Algebra via {Schr{\"o}dingerisation}}},
  journal       = {Proceedings of the Royal Society A: Mathematical, Physical and Engineering Sciences},
  volume        = {480},
  number        = {2291},
  pages         = {20230370},
  year          = {2024},
  doi           = {10.1098/rspa.2023.0370},
  eprint        = {2304.02865},
  archivePrefix = {arXiv},
  primaryClass  = {quant-ph}
}

@article{JinLiuYu2024,
  author        = {Jin, Shi and Liu, Nana and Yu, Yue},
  title         = {{Quantum Simulation of Partial Differential Equations via {Schr{\"o}dingerisation}}},
  journal       = {Physical Review Letters},
  volume        = {133},
  pages         = {230602},
  year          = {2024},
  doi           = {10.1103/PhysRevLett.133.230602},
  eprint        = {2212.13969},
  archivePrefix = {arXiv},
  primaryClass  = {quant-ph}
}

@book{BrennerScott2008,
  author        = {Brenner, Susanne C. and Scott, L. Ridgway},
  title         = {{The Mathematical Theory of Finite Element Methods}},
  series        = {Texts in Applied Mathematics},
  volume        = {15},
  edition       = {3},
  publisher     = {Springer},
  address       = {New York},
  year          = {2008},
  doi           = {10.1007/978-0-387-75934-0}
}

@book{Ciarlet1978,
  author        = {Ciarlet, Philippe G.},
  title         = {{The Finite Element Method for Elliptic Problems}},
  series        = {Studies in Mathematics and its Applications},
  volume        = {4},
  publisher     = {North-Holland},
  address       = {Amsterdam},
  year          = {1978}
}

@article{BerryChildsOstranderWang2017,
  author        = {Berry, Dominic W. and Childs, Andrew M. and Ostrander, Aaron and Wang, Guoming},
  title         = {{Quantum Algorithm for Linear Differential Equations with Exponentially Improved Dependence on Precision}},
  journal       = {Communications in Mathematical Physics},
  volume        = {356},
  number        = {3},
  pages         = {1057--1081},
  year          = {2017},
  doi           = {10.1007/s00220-017-3002-y},
  eprint        = {1701.03684},
  archivePrefix = {arXiv},
  primaryClass  = {quant-ph}
}

@article{Krovi2023,
  author        = {Krovi, Hari},
  title         = {{Improved Quantum Algorithms for Linear and Nonlinear Differential Equations}},
  journal       = {Quantum},
  volume        = {7},
  pages         = {913},
  year          = {2023},
  doi           = {10.22331/q-2023-02-02-913},
  eprint        = {2202.01054},
  archivePrefix = {arXiv},
  primaryClass  = {quant-ph}
}

@article{AnLiuLin2023,
  author        = {An, Dong and Liu, Jin-Peng and Lin, Lin},
  title         = {{Linear Combination of Hamiltonian Simulation for Nonunitary Dynamics with Optimal State Preparation Cost}},
  journal       = {Physical Review Letters},
  volume        = {131},
  pages         = {150603},
  year          = {2023},
  doi           = {10.1103/PhysRevLett.131.150603},
  eprint        = {2303.01029},
  archivePrefix = {arXiv},
  primaryClass  = {quant-ph}
}

@misc{AnChildsLin2023,
  author        = {An, Dong and Childs, Andrew M. and Lin, Lin},
  title         = {{Quantum Algorithm for Linear Non-Unitary Dynamics with Near-Optimal Dependence on All Parameters}},
  year          = {2023},
  eprint        = {2312.03916},
  archivePrefix = {arXiv},
  primaryClass  = {quant-ph}
}

@misc{AnChildsLinYing2024Laplace,
  author        = {An, Dong and Childs, Andrew M. and Lin, Lin and Ying, Lexing},
  title         = {{Laplace Transform Based Quantum Eigenvalue Transformation via Linear Combination of Hamiltonian Simulation}},
  year          = {2024},
  eprint        = {2411.04010},
  archivePrefix = {arXiv},
  primaryClass  = {quant-ph}
}

@article{CampsLinVanBeeumenYang2024,
  author  = {Camps, Daan and Lin, Lin and Van Beeumen, Roel and Yang, Chao},
  title   = {{Explicit Quantum Circuits for Block Encodings of Certain Sparse Matrices}},
  journal = {SIAM Journal on Matrix Analysis and Applications},
  volume  = {45},
  number  = {1},
  pages   = {801--827},
  year    = {2024},
  doi     = {10.1137/22M1484298}
}

@article{SunderhaufCampbellCamps2024,
  author        = {S{\"u}nderhauf, Christoph and Campbell, Earl and Camps, Joan},
  title         = {{Block-Encoding Structured Matrices for Data Input in Quantum Computing}},
  journal       = {Quantum},
  volume        = {8},
  pages         = {1226},
  year          = {2024},
  doi           = {10.22331/q-2024-01-11-1226},
  eprint        = {2302.10949},
  archivePrefix = {arXiv},
  primaryClass  = {quant-ph}
}

@article{ChildsLiuOstrander2021,
  author        = {Childs, Andrew M. and Liu, Jin-Peng and Ostrander, Aaron},
  title         = {{High-Precision Quantum Algorithms for Partial Differential Equations}},
  journal       = {Quantum},
  volume        = {5},
  pages         = {574},
  year          = {2021},
  doi           = {10.22331/q-2021-11-10-574},
  eprint        = {2002.07868},
  archivePrefix = {arXiv},
  primaryClass  = {quant-ph}
}

@article{TongAnWiebeLin2021,
  author        = {Tong, Yu and An, Dong and Wiebe, Nathan and Lin, Lin},
  title         = {{Fast Inversion, Preconditioned Quantum Linear System Solvers, Fast Green's-Function Computation, and Fast Evaluation of Matrix Functions}},
  journal       = {Physical Review A},
  volume        = {104},
  pages         = {032422},
  year          = {2021},
  doi           = {10.1103/PhysRevA.104.032422},
  eprint        = {2008.13295},
  archivePrefix = {arXiv},
  primaryClass  = {quant-ph}
}

@article{BrezziLipnikovShashkov2005,
  author  = {Brezzi, Franco and Lipnikov, Konstantin and Shashkov, Mikhail},
  title   = {{Convergence of the Mimetic Finite Difference Method for Diffusion Problems on Polyhedral Meshes}},
  journal = {SIAM Journal on Numerical Analysis},
  volume  = {43},
  number  = {5},
  pages   = {1872--1896},
  year    = {2005},
  doi     = {10.1137/040613950}
}

@article{HymanShashkovSteinberg1997,
  author  = {Hyman, James M. and Shashkov, Mikhail and Steinberg, Stanly},
  title   = {{The Numerical Solution of Diffusion Problems in Strongly Heterogeneous Non-Isotropic Materials}},
  journal = {Journal of Computational Physics},
  volume  = {132},
  number  = {1},
  pages   = {130--148},
  year    = {1997},
  doi     = {10.1006/jcph.1996.5633}
}

@misc{LowSomma2025Nonunitary,
  author        = {Low, Guang Hao and Somma, Rolando D.},
  title         = {{Optimal Quantum Simulation of Linear Non-Unitary Dynamics}},
  year          = {2025},
  eprint        = {2508.19238},
  archivePrefix = {arXiv},
  primaryClass  = {quant-ph}
}

@book{Saad2003Iterative,
  author        = {Saad, Yousef},
  title         = {{Iterative Methods for Sparse Linear Systems}},
  edition       = {2},
  publisher     = {Society for Industrial and Applied Mathematics},
  address       = {Philadelphia, PA},
  year          = {2003},
  doi           = {10.1137/1.9780898718003},
  isbn          = {978-0-898715-34-7}
}

@book{GolubVanLoan2013Matrix,
  author        = {Golub, Gene H. and Van Loan, Charles F.},
  title         = {{Matrix Computations}},
  edition       = {4},
  publisher     = {Johns Hopkins University Press},
  address       = {Baltimore, MD},
  year          = {2013},
  isbn          = {978-1-4214-0794-4}
}

@misc{Li2025BeyondCondition,
  author        = {Li, Jianqiang},
  title         = {{A New Quantum Linear System Algorithm Beyond the Condition Number and Its Application to Solving Multivariate Polynomial Systems}},
  year          = {2025},
  eprint        = {2510.05588},
  archivePrefix = {arXiv},
  primaryClass  = {quant-ph}
}

\appendix
\section{Proof of the ODE stability }\label{appA}
\begin{proof}[Proof of \cref{thm:first-order-relaxation}]

The residual block satisfies
\[
    \dot{\bm w}=M_h\bm w,
    \qquad
    \bm w=
    \begin{bmatrix}
    \bm r\\
    \bm s
    \end{bmatrix},
    \qquad
    M_h=
    \begin{bmatrix}
    0&-G_h^\dagger\\
    G_h&-I
    \end{bmatrix}.
\]
We prove the stability estimate by a Lyapunov functional argument. Let
\[
    A_h=G_h^\dagger G_h,
    \qquad
    K_h:=A_h^{-1}G_h^\dagger,
    \qquad
    P_h:=G_hA_h^{-1}G_h^\dagger .
\]
The discrete Poincar\'e inequality gives
\[
    \norm{G_h\bm v}\ge \gamma_0\norm{\bm v},
\]
with \(\gamma_0>0\) independent of \(h\). Hence \(A_h\) is positive definite. Moreover,
\(P_h\) is the orthogonal projector onto \(\operatorname{Range}(G_h)\), and
\[
    K_hG_h=I .
\]
Also,
\[
    \norm{K_h\bm s}
    \le
    \gamma_0^{-1}\norm{P_h\bm s}
    \le
    \gamma_0^{-1}\norm{\bm s}.
\]

Choose a constant \(\eta>0\), depending only on \(\gamma_0\), such that
\[
    0<\eta<\gamma_0,
    \qquad
    2(1-\eta)-\frac{\eta}{\gamma_0^2}>0 .
\]
Define
\[
    \mathcal E(\bm r,\bm s)
    =
    \norm{\bm r}^2+\norm{\bm s}^2
    -2\eta\,\operatorname{Re}\langle \bm r,K_h\bm s\rangle .
\]
Since
\[
    \left|\langle \bm r,K_h\bm s\rangle\right|
    \le
    \gamma_0^{-1}\norm{\bm r}\norm{\bm s},
\]
the functional \(\mathcal E\) is uniformly equivalent to
\(\norm{\bm r}^2+\norm{\bm s}^2\):
\[
    \left(1-\frac{\eta}{\gamma_0}\right)
    \bigl(\norm{\bm r}^2+\norm{\bm s}^2\bigr)
    \le
    \mathcal E(\bm r,\bm s)
    \le
    \left(1+\frac{\eta}{\gamma_0}\right)
    \bigl(\norm{\bm r}^2+\norm{\bm s}^2\bigr).
\]

We now use $\mathcal E$ as the Lyapunov functional and  differentiate  it along the dynamics
\[
    \dot{\bm r}=-G_h^\dagger\bm s,
    \qquad
    \dot{\bm s}=G_h\bm r-\bm s .
\]
First,
\[
    \frac{d}{dt}
    \left(
    \norm{\bm r}^2+\norm{\bm s}^2
    \right)
    =
    -2\norm{\bm s}^2 .
\]
Let
\[
    C(t):=\operatorname{Re}\langle \bm r(t),K_h\bm s(t)\rangle .
\]
Using \(K_hG_h=I\), we obtain
\[
\begin{aligned}
    \dot C(t)
    &=
    \operatorname{Re}\langle \dot{\bm r},K_h\bm s\rangle
    +
    \operatorname{Re}\langle \bm r,K_h\dot{\bm s}\rangle  \\
    &=
    -\operatorname{Re}\langle G_h^\dagger\bm s,A_h^{-1}G_h^\dagger\bm s\rangle
    +
    \operatorname{Re}\langle \bm r,K_h(G_h\bm r-\bm s)\rangle \\
    &=
    -\norm{P_h\bm s}^2
    +
    \norm{\bm r}^2
    -
    C(t).
\end{aligned}
\]
Therefore,
\[
\begin{aligned}
    \frac{d}{dt}\mathcal E
    &=
    -2\norm{\bm s}^2
    -2\eta\dot C(t)\\
    &=
    -2\eta\norm{\bm r}^2
    -2\norm{\bm s}^2
    +2\eta\norm{P_h\bm s}^2
    +2\eta C(t).
\end{aligned}
\]
Decompose
\[
    \bm s=P_h\bm s+(I-P_h)\bm s .
\]
Using
\[
    2|C(t)|
    \le
    \norm{\bm r}^2+\gamma_0^{-2}\norm{P_h\bm s}^2,
\]
we get
\[
\begin{aligned}
    \frac{d}{dt}\mathcal E
    &\le
    -\eta\norm{\bm r}^2
    -
    \left(
    2(1-\eta)-\frac{\eta}{\gamma_0^2}
    \right)
    \norm{P_h\bm s}^2
    -
    2\norm{(I-P_h)\bm s}^2 .
\end{aligned}
\]
By the choice of \(\eta\), the right-hand side is bounded above by
\[
    -c_0\bigl(\norm{\bm r}^2+\norm{\bm s}^2\bigr)
\]
for some \(c_0>0\) independent of \(h\). Since \(\mathcal E\) is uniformly equivalent to
\(\norm{\bm r}^2+\norm{\bm s}^2\), there exist constants \(C_{\rm st}\ge1\) and
\(c_{\rm st}>0\), independent of \(h\), such that
\[
    \norm{e^{tM_h}}
    \le
    C_{\rm st}e^{-c_{\rm st}t},
    \qquad t\ge0 .
\]
This proves the stability estimate and hence
\[
    \norm{\bm w(t)}
    \le
    C_{\rm st}e^{-c_{\rm st}t}\norm{\bm w(0)} .
\]

Since \(\bm w(t)\) is integrable in time, the accumulator
\[
    \bm x_\ast
    :=
    \bm x(0)+\int_0^\infty P_r\bm w(\tau)\,d\tau
\]
exists. Let \(\bm q(t)\) be defined by \eqref{eq:q-implicit}. The propagated identities give
\[
    \bm r(t)=\bm b_h-G_h^\dagger\bm q(t),
    \qquad
    \bm s(t)=G_h\bm x(t)-\bm q(t).
\]
Taking \(t\to\infty\) and using \(\bm r(t),\bm s(t)\to0\) gives
\[
    \bm b_h-G_h^\dagger\bm q_\ast=0,
    \qquad
    G_h\bm x_\ast-\bm q_\ast=0.
\]
Therefore
\[
    G_h^\dagger G_h\bm x_\ast=\bm b_h.
\]

For the tail bound, use the semigroup property:
\[
    \bm x_\ast-\bm x(t)
    =
    \int_t^\infty P_r\bm w(\tau)\,d\tau
    =
    P_r\int_0^\infty e^{sM_h}\bm w(t)\,ds .
\]
Since \(\norm{P_r}=1\), the stability estimate gives
\[
    \norm{\bm x_\ast-\bm x(t)}
    \le
    \int_0^\infty C_{\rm st}e^{-c_{\rm st}s}\,ds\,\norm{\bm w(t)}
    =
    \frac{C_{\rm st}}{c_{\rm st}}\norm{\bm w(t)}.
\]
\end{proof}

\section{Proof of the ODE-based QLSA complexity }\label{appB}
\begin{proof}[Proof of \cref{thm:quantum-first-order-ode}]
By Corollary~\ref{cor:Gh-block-encoding}, the accumulator generator
\[
    L_h=
    \begin{bmatrix}
    0&P_r\\
    0&M_h
    \end{bmatrix},
    \qquad
    M_h=
    \begin{bmatrix}
    0&-G_h^\dagger\\
    G_h&-I
    \end{bmatrix},
\]
has a block encoding with normalization
\[
    \alpha_L=\bigO(h^{-1}),
\]
because it consists only of \(G_h\), \(G_h^\dagger\), identity blocks, the bounded projection
\(P_r\), and bounded damping/projector blocks. By \cref{thm:first-order-relaxation}, choosing
\[
    T
    =
    \widetildeO\!\left(\log\frac1\eps\right)
\]
makes the finite-time relaxation error \(\bigO(\eps)\), with logarithmic dependence on \(h\) and on
\(\norm{\bm b_h}/\norm{\bm x_\ast}\) absorbed in \(\widetildeO\). Hence
\[
    \alpha_L T
    =
    \widetildeO\!\left(
        h^{-1}\log\frac1\eps
    \right).
\]

We next account for the output normalization. The accumulator ODE is homogeneous:
\[
    \bm z(T)=e^{TL_h}\bm z(0),
    \qquad
    \bm z(0)=
    \begin{bmatrix}
    0\\
    \bm b_h\\
    0
    \end{bmatrix}.
\]
For LCHS-type algorithms, the state-preparation overhead is governed by the output-norm ratio of
the nonunitary evolution. In the homogeneous time-independent case, this is the factor
\[
    \Gamma_{\rm out}(T)
    :=
    \frac{\norm{\bm z(0)}}{\norm{\bm z(T)}} .
\]
This is the homogeneous specialization of the standard LCHS output-normalization factor; compare
the homogeneous-case corollaries of \cite{AnChildsLin2023}, where the state-preparation cost is
optimal in this ratio.

In the present accumulator system,
\[
    \norm{\bm z(0)}=\norm{\bm b_h},
    \qquad
    \norm{\bm z(T)}^2
    =
    \norm{\bm x(T)}^2+\norm{\bm w(T)}^2 .
\]
By the choice of \(T\),
\[
    \norm{\bm x(T)-\bm x_\ast}
    =
    \bigO(\eps)\norm{\bm x_\ast},
    \qquad
    \norm{\bm w(T)}
    =
    \bigO(\eps)\norm{\bm x_\ast}.
\]
Therefore
\[
    \norm{\bm z(T)}
    =
    (1+\bigO(\eps))\norm{\bm x_\ast},
\]
and hence
\[
    \Gamma_{\rm out}(T)
    =
    (1+\bigO(\eps))
    \frac{\norm{\bm b_h}}{\norm{\bm x_\ast}} .
\]
If the right-hand side is normalized before state preparation, then
\[
    \widehat{\bm x}_\ast
    =
    A_h^{-1}\ket{\bm b_h}
    =
    \frac{\bm x_\ast}{\norm{\bm b_h}},
\]
so the same factor is
\[
    \Gamma_{\rm out}(T)
    =
    (1+\bigO(\eps))
    \frac{1}{\norm{A_h^{-1}\ket{\bm b_h}}}.
\]

Combining the LCHS query dependence on \(\alpha_L T\), the input-state preparation cost \(C_b\),
and the output-normalization factor \(\Gamma_{\rm out}(T)\) gives
\[
    \widetildeO\!\left(
        \frac{\norm{\bm b_h}}{\norm{\bm x_\ast}}
        \left[
        C_b+
        h^{-1}\polylog\frac1\eps
        \right]
    \right).
\]

It remains to check the cost of extracting the solution block. Since
\[
    \bm z(T)=
    \begin{bmatrix}
    \bm x(T)\\
    \bm w(T)
    \end{bmatrix},
\]
the probability of selecting the \(\bm x\)-block is
\[
    p_x(T)
    =
    \frac{\norm{\bm x(T)}^2}
    {\norm{\bm x(T)}^2+\norm{\bm w(T)}^2}.
\]
Using the bounds above,
\[
    \norm{\bm x(T)}
    =
    (1+\bigO(\eps))\norm{\bm x_\ast},
    \qquad
    \norm{\bm w(T)}
    =
    \bigO(\eps)\norm{\bm x_\ast},
\]
so
\[
    p_x(T)=1-\bigO(\eps^2).
\]
Thus selecting the \(\bm x\)-block succeeds with order-one probability and does not change the
leading mesh dependence. Finally, since
\[
    \norm{\bm x(T)-\bm x_\ast}
    =
    \bigO(\eps)\norm{\bm x_\ast},
\]
the normalized states \(\ket{\bm x(T)}\) and \(\ket{\bm x_\ast}\) are \(\bigO(\eps)\)-close.
\end{proof}

\end{document}